\documentclass[12pt, draftclsnofoot, onecolumn]{IEEEtran}
\usepackage{amsmath, mathrsfs, mathtools}
\usepackage{amsfonts}
\usepackage{amssymb}
\usepackage{color}
\usepackage{multirow}
\usepackage{stmaryrd}
\allowdisplaybreaks
\usepackage{caption}
\usepackage{subcaption}
\usepackage{float}
\usepackage{amsfonts}
\usepackage{amssymb}
\usepackage{color}
\usepackage{multirow}
\usepackage{graphicx}
\usepackage{tcolorbox}
\usepackage[numbers,sort&compress]{natbib}

\usepackage{geometry}
\allowdisplaybreaks

\newcommand{\hmH}{h_{m,{\sf H}}}
\newcommand{\hmV}{h_{m,{\sf V}}}
\newcommand{\Sigmah}{{\Sigmam}_{\hv}}
\newcommand{\Sigmay}{{\Sigmay}_{\yv}}
\newcommand{\WH}{W_{\sf H}}
\newcommand{\WV}{W_{\sf V}}
\newcommand{\gammaH}{\gamma_{\sf H}}
\newcommand{\gammaV}{\gamma_{\sf V}}
\def\Nvar{N_0}
\def\Pdl{P_{dl}}
\def\Tdl{T_{dl}}

\def\rect{{\ttr\tte\ttc\ttt}}

\def\wh{\widehat}


\usepackage{mathbbol}

\def\mindex#1{\index{#1}}

\def\sq{\hbox{\rlap{$\sqcap$}$\sqcup$}}
\def\qed{\ifmmode\sq\else{\unskip\nobreak\hfil
\penalty50\hskip1em\null\nobreak\hfil\sq
\parfillskip=0pt\finalhyphendemerits=0\endgraf}\fi\medskip}

\long\def\defbox#1{\framebox[.9\hsize][c]{\parbox{.85\hsize}{%
\parindent=0pt
\baselineskip=12pt plus .1pt      
\parskip=6pt plus 1.5pt minus 1pt 
 #1}}}

\long\def\beginbox#1\endbox{\subsection*{}%
\hbox{\hspace{.05\hsize}\defbox{\medskip#1\bigskip}}%
\subsection*{}}

\def\endbox{}

\def\rank{{\rm rank\,}}

\newsavebox{\junk}
\savebox{\junk}[1.6mm]{\hbox{$|\!|\!|$}}

\def\bC{{\mathbb C}}

\def\bE{{\mathbb E}}

\def\bR{{\mathbb R}}
\def\bS{{\mathbb S}}

\def\bfB{{\bf B}}

\def\bfU{{\bf U}}

\def\bfy{{\bf y}}
\def\bfz{{\bf z}}

\def\ttc{{\mathtt c}}

\def\tte{{\mathtt e}}

\def\ttr{{\mathtt r}}

\def\ttt{{\mathtt t}}

\def\sfH{{\sf H}}

\def\bfmath#1{{\mathchoice{\mbox{\boldmath$#1$}}%
{\mbox{\boldmath$#1$}}%
{\mbox{\boldmath$\scriptstyle#1$}}%
{\mbox{\boldmath$\scriptscriptstyle#1$}}}}

\def\bfmY{\bfmath{Y}}

\def\bfmhhaY{\bfmath{\hhaY}} 
\def\bfmhhaY{\hbox to 0pt{$\widehat{\bfmY}$\hss}\widehat{\phantom{\raise 1.25pt\hbox{$\bfmY$}}}}

\def\til={{\widetilde =}}

\def\clC{{\cal C}}

\def\clG{{\cal G}}

\def\clN{{\cal N}}

 \def\FRAC#1#2#3{\genfrac{}{}{}{#1}{#2}{#3}}

\def\ddtp{{\mathchoice{\FRAC{1}{d^{\hbox to 2pt{\rm\tiny +\hss}}}{dt}}%
{\FRAC{1}{d^{\hbox to 2pt{\rm\tiny +\hss}}}{dt}}%
{\FRAC{3}{d^{\hbox to 2pt{\rm\tiny +\hss}}}{dt}}%
{\FRAC{3}{d^{\hbox to 2pt{\rm\tiny +\hss}}}{dt}}}}

\def\average#1,#2,{{1\over #2} \sum_{#1}^{#2}}

\def\eye(#1){{\bf(#1)}\quad}

\newtheorem{theorem}{{\bf Theorem}}

\newtheorem{corollary}{Corollary}

\newtheorem{remark}{{\bf Remark}}

\newtheorem{lemma}[theorem]{{\bf Lemma}}

\def\eq#1/{(\ref{e:#1})}

\newcommand{\beqn}[1]{\notes{#1}%
\begin{eqnarray} \elabel{#1}}

\newcommand{\eeqn}{\end{eqnarray} }

\newcommand{\beq}[1]{\notes{#1}%
\begin{equation}\elabel{#1}}

\newcommand{\eeq}{\end{equation}}

\def\bdes{\begin{description}}
\def\edes{\end{description}}

\newcounter{rmnum}

\newcounter{anum}

{\end{list}}

\def\ass(#1:#2){(#1\ref{#1:#2})}

\def\ritem#1{
\item[{\sf \ass(\current_model:#1)}]
}

\newenvironment{recall-ass}[1]{%
\begin{description}
\def\current_model{#1}}{
\end{description}
}

\def\herm{{\sfH}}

\def\cg{{\clC\clN}} 
\newcommand{\normd}[1]{{\left\vert\kern-0.25ex\left\vert\kern-0.25ex\left\vert #1 
		\right\vert\kern-0.25ex\right\vert\kern-0.25ex\right\vert}}

\newcommand{\hul}{\hv_{ul}}

\newcommand{\hdlk}{\hv_{dl,k}}

\long\def\comment#1{}

\newcommand{\av}{{\bf a}}
\newcommand{\bv}{{\bf b}}

\newcommand{\dv}{{\bf d}}

\newcommand{\fv}{{\bf f}}
\newcommand{\gv}{{\bf g}}
\newcommand{\hv}{{\bf h}}

\newcommand{\sv}{{\bf s}}

\newcommand{\uv}{{\bf u}}
\newcommand{\wv}{{\bf w}}

\newcommand{\xv}{{\bf x}}
\newcommand{\yv}{{\bf y}}
\newcommand{\zv}{{\bf z}}

\newcommand{\Am}{{\bf A}}
\newcommand{\Bm}{{\bf B}}
\newcommand{\Cm}{{\bf C}}
\newcommand{\Dm}{{\bf D}}

\newcommand{\Fm}{{\bf F}}
\newcommand{\Gm}{{\bf G}}
\newcommand{\Hm}{{\bf H}}
\newcommand{\Id}{{\bf I}}
\newcommand{\Jm}{{\bf J}}

\newcommand{\Lm}{{\bf L}}

\newcommand{\Om}{{\bf O}}
\newcommand{\Pm}{{\bf P}}
\newcommand{\Qm}{{\bf Q}}
\newcommand{\Rm}{{\bf R}}
\newcommand{\Sm}{{\bf S}}
\newcommand{\Tm}{{\bf T}}
\newcommand{\Um}{{\bf U}}
\newcommand{\Wm}{{\bf W}}
\newcommand{\Vm}{{\bf V}}
\newcommand{\Xm}{{\bf X}}
\newcommand{\Ym}{{\bf Y}}

\newcommand{\Ac}{{\cal A}}

\newcommand{\Cc}{{\cal C}}

\newcommand{\Ec}{{\cal E}}

\newcommand{\Gc}{{\cal G}}

\newcommand{\Ic}{{\cal I}}
\newcommand{\Jc}{{\cal J}}
\newcommand{\Kc}{{\cal K}}

\newcommand{\Mc}{{\cal M}}
\newcommand{\Nc}{{\cal N}}
\newcommand{\Oc}{{\cal O}}

\newcommand{\Sc}{{\cal S}}

\newcommand{\Vc}{{\cal V}}

\newcommand{\lambdav}{\hbox{\boldmath$\lambda$}}

\newcommand{\Gammam}{\hbox{\boldmath$\Gamma$}}

\newcommand{\Sigmam}{\hbox{\boldmath$\Sigma$}}
\newcommand{\Phim}{\hbox{\boldmath$\Phi$}}

\newcommand{\Psim}{\hbox{\boldmath$\Psi$}}

\newcommand{\trace}{{\hbox{tr}}}
\renewcommand{\arg}{{\hbox{arg}}}

\newcommand{\transp}{{\sf T}}

\title{Dual-Polarized FDD Massive MIMO: \\
A Comprehensive Framework}
\author{Mahdi Barzegar Khalilsarai$^1$,  Tianyu Yang$^1$, Saeid Haghighatshoar$^2$, Xinping Yi$^3$, and Giuseppe Caire$^1$  
\thanks{$^1$ Communications and Information Theory Group (CommIT), Technische Universit\"{a}t Berlin (\{m.barzegarkhalilsarai, tianyu.yang,  caire\}@tu-berlin.de).}
\thanks{$^2$ Saeid Haghighatshoar is currently with the Swiss Center for Electronics and Microtechnology (CSEM), however his contribution to this work was made while he was with the CommIT group (saeid.haghighatshoar@csem.ch).}
\thanks{$^3$ Department of Electrical Engineering and Electronics, University of Liverpool  (xinping.yi@liverpool.ac.uk).}
}

\begin{document}
\maketitle
\vspace{-7mm}
\begin{abstract}
We propose a comprehensive scheme for realizing a massive multiple-input multiple-output (MIMO) system with dual-polarized antennas in frequency division duplexing (FDD) mode. Employing dual-polarized elements in a massive MIMO array has been common practice recently and can, in principle, double the number of spatial degrees of freedom with a less-than-proportional increase in array size. However, processing a dual-polarized channel is demanding due to the high channel dimension and the lack of Uplink-Downlink (UL-DL) channel reciprocity in FDD mode. In particular, the difficulty arises in channel covariance acquisition for both UL and DL transmissions and in common training of DL channels in a multi-user setup. To overcome these challenges, we develop a unified framework consisting of three steps: (1) a covariance estimation method to efficiently estimate the UL covariance from noisy, orthogonal UL pilots; (2) a UL-DL covariance transformation method that obtains the DL covariance from the estimated UL covariance in the previous step; (3) a multi-user common DL channel training with limited DL pilot dimension method, which enables the BS to estimate effective user DL channels and use them for interference-free DL beamforming and data transmission. We provide extensive empirical results to prove the applicability and merits of our scheme.\footnote{The methodology and results of this work were partially published in two recent conference articles by the same authors \cite{khalilsarai2019structured,khalilsarai2020active}.}

\end{abstract}

\section{Introduction}

\subsection{Dual Polarized FDD massive MIMO}
Massive multiple-input multiple-output (MIMO) antenna systems promise high data rates as well as link reliability in prospective generations of wireless communication systems \cite{larsson2014massive,lu2014overview}. The characteristic property of these systems is the deployment of a large number ($M\gg 1$) of antennas at the base station (BS), resulting in substantial improvements in terms of beamforming and multiplexing gains, while also increasing the array size. Since most wireless networks are currently based on \textit{frequency division duplexing} (FDD), implementing a massive MIMO system in FDD mode is an appealing proposition. Besides, many network developers consider using dual-polarized (DP) antenna elements in the array, since it offers a doubling of the number of inputs with a less-than-proportional increase in array size \cite{degli2011analysis,xu2016dual}. The effect of adopting DP antennas at the array on performance metrics such as the multiplexing gain depends on the degree of co-polarization (co-pol) and cross-polarization (X-pol) between the two polarization states (namely, \textit{horizontal} (H) and \textit{vertical} (V) polarizations). While specular reflection components lead to a low degree of X-pol (hence an approximate decoupling of the polarizations), diffuse scattering results in relatively high X-pol \cite{degli2011analysis}. 
In order to study these effects, we assign a pair of (correlated) channel coefficients to each element of the array and introduce a statistical model to represent the co-pol and X-pol properties of a particular environment. 
Assuming Gaussian statistics, the channel is a $2M$-dimensional random vector that is statistically characterized by its mean and covariance. This doubling of dimension brings about a series of challenges in realizing a dual-polarized FDD massive MIMO system. In the follow-up to this section, we outline these challenges and explain our proposed treatment for tackling each. 
\subsection{Channel Covariance Estimation}
Channel covariance knowledge at the BS is crucial for a variety of tasks including minimum mean squared error (MMSE) channel estimation, user grouping and designing efficient DL precoders. During UL, each user transmits a number of orthogonal pilots to the BS. Ideally, these pilots are separated by the time-frequency channel coherence block, so that with each transmission an independent realization of the channel $\hul (i)\in \bC^{2M} ,\, i=1,\ldots,N$ is observed at the BS.  The BS in turn uses the set of observed channel samples to estimate the UL channel covariance. The simplest and most common estimator is the sample covariance $\wh{\Sigmam}_N = \frac{1}{N}\sum_{i=1}^{N} \hul (i)\hul (i)^\herm,$ which is an unbiased estimator of the true covariance $\Sigmam = \bE [\xv \xv^\herm ]$. It is well-known that in scenarios in which the number of samples ($N$) is in the order of signal dimension ($N=\Oc (2M)$ where $\Oc$ denotes the Big O notation), the sample covariance estimator can be substantially improved by exploiting covariance structure. This is precisely the case when we study DP massive MIMO channels, in which the channel dimension is high ($2M\gg 1$) and the number of samples is restricted by the number of available time-frequency pilot resources.

The idea of exploiting structure for the purpose of covariance estimation is not new. Recent interest in low-rank and sparse covariance models has given rise to a broad range of such methods. The common denominator of these estimators is to form an optimization problem with the covariance estimate as its variable, in which a suitable cost, corresponding to the structure is minimized. For example, methods based on rank minimization (for low-rank covariances), and $\ell_0$-pseudo-norm minimization (for sparse covariances) or a combination thereof are proposed \cite{kang2014rank,oymak2015simultaneously}. Alternatively, one may consider convex relaxations of the costs above, replacing the matrix rank with its nuclear norm and the $\ell_0$-pseudo-norm with $\ell_1$-norm \cite{luo2011recovering}. Several interesting variations of this idea exist but going into further details is out of the scope of this work (see, for example, \cite{ravikumar2011high}).

In order to exploit structure in the problem at hand, we will show that the DP channel covariance follows a Kronecker-type form, and is given by an integral transform involving a positive semidefinite matrix-valued function of the angle of arrival (AoA). This function, coined as the \textit{dual-polarized angular spread function} (DP-ASF), represents the channel angular power density in H and V polarizations as well as the cross-correlation between the two.  Our approach to covariance estimation is based on a parametric representation of the DP-ASF in terms of a linear combination of elementary, limited-support density functions, whose coefficients are estimated given independent DP channel samples $\{ \hul (i) \}_{i=1}^N$. This parametric model is general, in that, it incorporates specular as well as diffuse angular scattering and does not assume unverified polarization properties. The estimation is carried out via a convex program, which enforces the positive semidefinite property on the solution. After estimating the DP-ASF, an estimate of the covariance is readily given by a simple integral transform. 
\subsection{Uplink-Downlink Channel Covariance Transformation}
In addition to the UL covariance, the BS needs to obtain an estimate of the DL covariance for all users both to obtain a reliable estimate of user DL channels and to design a DL precoder for multi-user beamforming. In a \textit{time division duplexing} (TDD) system, the covariance is the same for Uplink (UL) and Downlink (DL) channels due to channel reciprocity \cite{marzetta2006much}. However, in an FDD system, UL and DL covariances are different and therefore the DL covariance has to be estimated. In addition, the overhead of transmitting DL pilots, receiving feedback from the users and then estimating the DL covariance is too large and therefore this is not a feasible process. In order to estimate the DL covariance, we propose a UL-DL covariance transformation method. This method hinges upon a phenomenon known as \textit{angular channel reciprocity}: the angular power density as seen from the array is the same for UL and DL, resulting in the DP-ASF to be identical during UL and DL. The concept of angular channel reciprocity is well-established in the literature (e.g. \cite{khalilsarai2018fdd,miretti2018fdd}) and is exploited for processing the single-polarized array. Having an estimate of the DP-ASF from the previous step, we use angular reciprocity to obtain an estimate of the DL covariance. The relation between the DP-ASF and the DL covariance is similar to that of the DP-ASF and the UL covariance, with a change of the array response vector due to the change of frequency band. 
\begin{remark}
	We emphasize that, we exploit UL-DL angular reciprocity to estimate the DL covariance, which is then used to design the sparsifying precoder, and allow the estimation of instantaneous DL channels via common DL training and channel state feedback with limited pilot dimension. In contrast, some works in the literature have proposed to use UL pilots to directly estimate the instantaneous DL channels. Such an approach is reasonable only in a highly optimistic case in which the channel coefficient per antenna is seen as a superposition of signals coming from discrete, separable paths, whose AoA and complex coefficients can be estimated via super-resolution or compressive sensing methods.  These techniques fail in the presence of diffuse scattering components, where signal paths are not separable and ``extrapolating" the instantaneous UL channel to the instantaneous DL channel results in an MMSE proportional to the amount of power coming from diffuse scattering \cite{khalilsarai2017compressive}.\hfill  $\lozenge$
\end{remark}

\subsection{Downlink Channel Training and Precoding via Active Sparsification}
In order to achieve the gains of massive MIMO, it is necessary for the BS to estimate (train) instantaneous user DL channels and perform interference-free DL beamforming. While channel training is an easy task with small MIMO arrays, it becomes increasingly challenging with the increase in the number of antennas. This is especially an issue in FDD mode, where, unlike the TDD mode, instantaneous channel reciprocity does not hold, and UL and DL channels corresponding to different frequency bands are virtually uncorrelated (and therefore statistically independent, due to Gaussianity) random vectors, whose statistics is related by the UL-DL covariance relationship explained earlier. Therefore, the BS has to probe the channel in the DL by broadcasting pilot symbols, receive feedback from the users and finally estimate the DL channel. In order to estimate a $(2M\gg 1)$-dimensional DP channel with any conventional method and without structural assumptions (such as channel sparsity), the BS needs to transmit at least $2M$ pilot symbols and receive their feedback in the UL. On the other hand, the time-frequency resources of a single coherence block are used for both channel training and data transmission. Dedicating a number $T_{dl}$ of a total of $T$ coherence block dimensions to DL training introduces a \textit{pre-log factor} of $\max \{0,1-T_{dl}/T  \}$ in the sum-rate. Conventional channel estimation requires $T_{dl}\ge 2M$, while the dimension $T$ may in fact be less than $2M$. For example, in a standard LTE setup the users are scheduled over resource blocks containing $14$ OFDM symbols and $12$ subcarriers, making a total of $T = 14\times 12=168$ dimensions \cite{sesia2011lte}. With a DP array of, say, $M=100$ antennas, the number of coefficients to be estimated amounts to $2M = 200>T$, which clearly exhausts all the resources and results in zero sum-rate. This problem is not solved even by resorting to the channel sparsity assumption and various compressed sensing (CS) techniques (see e.g. \cite{rao2014distributed} and \cite{ding2018dictionary}). First, the channel sparsity postulate may not be always verified: sources such as \cite{bjornson2016massive} call it the \textit{sparsity hypothesis}. Therefore, CS techniques are always at the mercy of environmental properties, as to whether the channel is indeed sparse or not. Second, even if the sparsity assumption holds, the number of measurements necessary for accurate sparse recovery might be still high, exceeding the available DL pilot dimension.

To resolve this issue, we adopt and extend the \textit{active channel sparsification} (ACS) approach first proposed by some of the authors in \cite{khalilsarai2018fdd} for single-polarized arrays and extended here to DP arrays. Given user DL covariances and for a given pilot dimension $T_{dl}$, the idea behind ACS is to design a sparsifying precoder that jointly reduces the number of significant angular components of all the user channels to less than $T_{dl}$, while at the same time trying to maximize the rank of the sparsified effective channel matrix. This enables, as it will be shown, stable recovery of the effective user channels and simultaneously maximizing the system multiplexing gain, which is proportional to the channel matrix rank. Using the ACS method, we are not at the mercy of channel's sparsity features and we do not make any assumptions thereof. ACS is deployed via first identifying a set of common virtual beams among all the users for channel representation and forming a user-virtual beam bipartite graph. Then we prove a result, relating the channel matrix rank to the maximal matching size in the graph. Finally, the sparsifying precoder is realized by selecting a subset of users and virtual beams as the solution to a \textit{mixed integer linear program} (MILP) which can be solved via standard methods for practical channel dimensions. 

\subsection{Organization} 
The paper is organized as follows. In Section \ref{sec:model} we introduce the dual-polarized channel model. In Section \ref{sec:cov_estimation} we develop our channel covariance estimation method. Section \ref{sec:transformation} discusses UL-DL covariance transformation. In Section \ref{sec:common_training} we introduce the user-virtual beam bipartite graph and explain the ACS method. Various empirical results in Section \ref{sec:simulations} conclude the paper. 

\section{Channel Model}\label{sec:model}
We consider a \textit{uniform linear array} (ULA) of $M$ dual-polarized antenna elements that communicates with a single-antenna, single-polarized user. The channel between antenna $m$ of the array and the user antenna consists of two elements, corresponding to horizontal (H) and vertical (V) polarization coefficients, respectively denoted as $\hmH^{ul},\, \hmV^{ul}\in \bC$ for the Uplink (UL) channel. The channel gain for either polarization is a superposition of random gains along a continuum of AoAs, weighted by the \textit{antenna element response} which for antenna $m$ is given by  $a_m = e^{j\pi m\frac{2 d \sin (\theta)}{\lambda_{ul}}} = e^{j\pi m\frac{ \sin (\theta)}{\sin (\theta_{\max})}}$, where $d$ is the antenna spacing, $\theta\in [-\theta_{\max},\theta_{\max}]$ is the AoA, $\theta_{\max}$ is the maximum array angular aperture and $\lambda_{ul}$ is the wave-length of the electromagnetic wave over the UL frequency band. Taking the antenna spacing to be $d=\frac{\lambda_{ul}}{2\sin \theta_{\max}}$ and with the change of variables $\xi=\frac{\sin \theta}{\sin \theta_{\max}}\in [-1,1]$, the antenna element response admits the simpler form $a_m = e^{jm\pi \xi},~m=0,\ldots,M-1,$ with $\xi$ denoting the ``normalized" AoA parameter. Then, one can express H and V channel coefficients as
\begin{equation}\label{eq:h_expressions}
\begin{aligned}
\hmH^{ul} & = \int_{-1}^1 \WH (\xi) e^{j\pi m \xi} d\xi,  \quad
\hmV^{ul}  = \int_{-1}^1 \WV (\xi) e^{j\pi m\xi} d\xi  \,
\end{aligned}
\end{equation} 
where $\WH$ and $\WV$ are random processes representing the random gains along each AoA for H and V polarizations, respectively. We assume both of these to be zero-mean, circularly symmetric, complex Gaussian processes with the following autocorrelations:
\begin{equation}\label{eq:autocorr}
\begin{aligned}
\bE \left[ \WH (\xi) \WH^\ast (\xi') \right] = \gammaH (\xi) \delta (\xi -\xi'),\quad
\bE \left[ \WV (\xi) \WV^\ast  (\xi') \right] = \gammaV (\xi) \delta (\xi -\xi'),
\end{aligned}
\end{equation}
where we have adopted the \textit{wide-sense stationary uncorrelated scattering} (WSSUS) model, which assumes stationary second-order channel statistics (over reasonably short time intervals) and uncorrelated angular scattering gains \cite{proakis2001digital}. The functions\footnote{We use the term ``function" with some abuse of terminology. An accurate term would be ``distribution" in the sense of generalized functions, as studied in \cite{gelfand1964generalized}.} $\gammaH$ and $\gammaV$ are both real and non-negative, representing the channel power density received along each AoA for H and V polarizations, respectively. We call these horizontal and vertical \textit{angular spread functions} (ASFs) (see Fig. \ref{fig:DPASF}).
\begin{figure}[t]
	\centering
	\includegraphics[ width=0.5\textwidth]{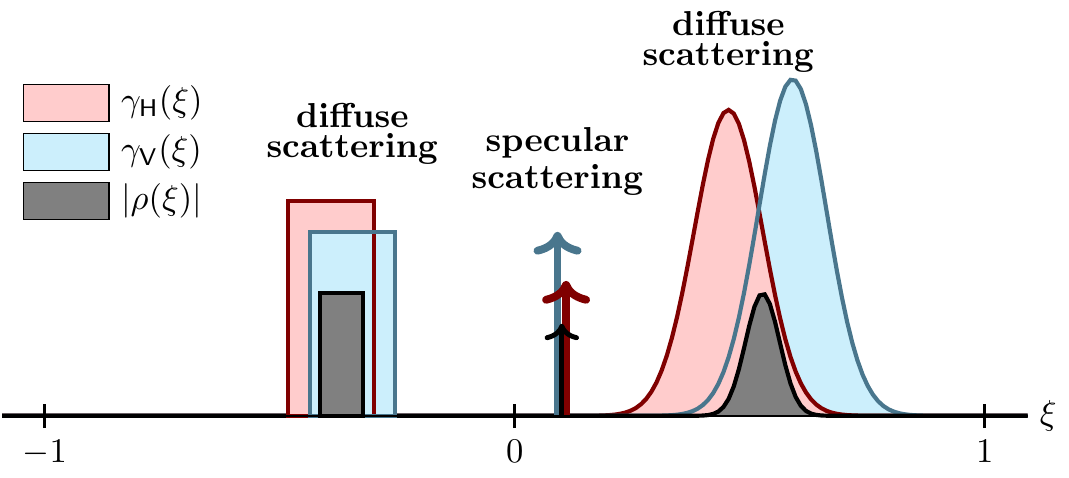}
	\caption{\small An example of H and V ASFs as well as the H-V cross-correlation modulus. The blue shaded function highlights $\gammaV (\xi)$, the red one highlights $\gammaH (\xi)$ and the black one highlights $|\rho (\xi)|$. }
	\label{fig:DPASF}
\end{figure} 
In practice, the H and V links can \textit{not} be entirely isolated from each other and therefore, there exists a leakage of channel power between the two. This implies that, for each AoA, the random gains $\WH (\xi)$ and $\WV (\xi)$ are correlated such that we have $\bE \left[ \WH (\xi) \WV^\ast  (\xi') \right] = \rho (\xi) \delta (\xi -\xi'),$
where $\rho$ is a generally complex-valued function. 

A dual-polarized channel can be more conveniently expressed as follows. Denote $M$-dimensional horizontal and vertical UL channel vectors 
$\hv_{\sf H}^{ul}  = \left[h_{0,{\sf H}},\ldots,h_{M-1,{\sf H}} \right]^\transp$, 
$\hv_{\sf V}^{ul}  = \left[h_{0,{\sf V}},\ldots,h_{M-1,{\sf V}} \right]^\transp$ and define the dual-polarized channel via the $2M$-dimensional vector $\hul = [ \hv_{\sf H}^{ul\, \transp},\allowbreak \, \hv_{\sf V}^{ul\, \transp} ]^\transp $. Using \eqref{eq:h_expressions} we have
$\hv_{\sf H}^{ul} = \int_{-1}^1 W_{\sf H} (\xi) \av (\xi) d\xi$ and 
$\hv_{\sf V}^{ul} = \int_{-1}^1 W_{\sf V} (\xi) \av (\xi) d\xi,$
where $\av (\xi ) = [1,\,e^{j\pi \xi},\ldots,\allowbreak e^{j\pi (M-1)\xi}]^\transp \in \bC^{M}$ denotes the array response vector. Finally, the DP channel is given by
\begin{equation}\label{eq:h_dual}
\begin{aligned}
\hv_{ul}= \int_{-1}^1 \begin{bmatrix}
\av (\xi) & \mathbf{0}\\
\mathbf{0} & \av (\xi) 
\end{bmatrix}\, 
\begin{bmatrix}
W_{\sf H} (\xi)\\
W_{\sf V}  (\xi) 
\end{bmatrix} \, d\xi 
 = \int_{-1}^1 \left(\mathbf{I}_2 \otimes \av (\xi)\right) \wv (\xi) d\xi,
\end{aligned}
\end{equation}
where $\otimes $ denotes Kronecker product, and $\wv (\xi) := [W_{\sf H} (\xi),W_{\sf V} (\xi)]^\transp$. The channel covariance can be computed according to \eqref{eq:h_dual} as
\begin{equation}\label{eq:ch_cov}
\begin{aligned}
\scalebox{0.95}{$\Sigmam_{\hv}^{ul}  = \bE \left[ \hv_{ul} \hv_{ul}^\herm \right] =\int_{-1}^1 \int_{-1}^1 \left(\mathbf{I}_2 \otimes \av (\xi)\right)  \bE \left[ \wv (\xi)\wv (\xi')^\herm \right] \left(\mathbf{I}_2 \otimes \av (\xi')\right)^\herm d\xi d\xi'  = \int_{-1}^1 \Gammam (\xi) \otimes \Am_{ul} (\xi) d\xi$},
\end{aligned}
\end{equation} 
where we have defined the rank-1 matrix $\Am_{ul} (\xi) = \av (\xi) \av (\xi)^\herm$,
and the matrix-valued function
\begin{equation}\label{eq:DP_ASF}
\Gammam (\xi) = \bE \left[ \wv (\xi)\wv (\xi')^\herm \right]  = \begin{bmatrix}
\gammaH (\xi) & \rho (\xi) \\
\rho (\xi)^\ast &  \gammaV (\xi)
\end{bmatrix} \in \bC^{2\times 2},
\end{equation}
which is positive semidefinite (PSD) for all $\xi \in [-1,1]$. For convenience, we call $\Gammam (\xi) $ the \textit{dual-polarized angular spread function} (DP-ASF) and note that, similar to the role played by the ASF in a single-polarized array, the DP-ASF captures the angular spectral properties of the channel, i.e. the power density along H and V links and the power leakage density between the two. Note that since $\Gammam (\xi)$ is PSD, we have
$|\rho (\xi)|^2\le \gammaH (\xi) \gammaV (\xi)$,
for all $\xi \in [-1,1]$, putting a bound on the modulus of $\rho $. In particular, if for some $\xi$ we have $\gammaH (\xi)=0$ or $\gammaV (\xi)=0$, then necessarily $\rho (\xi)=0$, which shows that the support of $\rho $ is limited to the support of $\gammaV $ and $\gammaH$.

\section{Channel Covariance Estimation}\label{sec:cov_estimation}
Suppose that the BS receives $N$ noisy pilot measurements of the UL channel as
\begin{equation}\label{eq:noisy_pilots}
\yv_{ul} (i) = \hv_{ul} (i) \, x_n + \zv (i),~i=1,\ldots,N,
\end{equation}
where $x_n=\sqrt{P}$ is the pilot symbol, $\zv (i)\sim \cg (\mathbf{0},N_0\mathbf{I}_{2M})$ is the additive white Gaussian noise (AWGN) vector at the $i$-th transmission with $N_0$ being the noise variance per element, and $\hv_{ul} (i)$ is the $i$-th channel realization. With orthogonal pilot transmission over distinct time-frequency coherence blocks, we can safely assume that the channel realizations $\hv (i), \, \, n=1,\ldots,N$ are independent. A simple estimator of the UL channel covariance $\Sigmam_{\hv}^{ul}$ is given by the sample covariance matrix
\begin{equation}\label{eq:sample_cov}
\widehat{\Sigmam}_{\hv}^{ul} = \widehat{\Sigmam}_{\yv}^{ul}-N_0 \mathbf{I}_{2M}:=\frac{1}{N}\sum_{i=1}^{N} \yv_{ul} (i)\yv_{ul} (i)^\herm -N_0 \mathbf{I}_{2M},
\end{equation}
The sample covariance is a consistent estimator of the true covariance and converges to it for relatively large number of samples ($N\gg 2M$), obtaining which is affordable in the case of small MIMO channels. However, for a dual-polarized massive MIMO channel with $2M\gg 1$, this condition is hardly met and instead, the number of samples is in the order of the channel dimension ($N=\Oc (2M)$). In these regimes of dimensionality, it is well-known that one can considerably improve the sample covariance estimator, for example by exploiting the covariance structure. In particular, here we are interested in covariance matrices that belong to the set of feasible DP MIMO covariances of a ULA defined as
\begin{equation}
\Cc := \left\{ \int_{-1}^{1} \Phim (\xi)\otimes \Am (\xi) d\xi,~ \Phim : [-1,1]\to \bS_+^{2} \right\},
\end{equation}
where $\Phim (\xi)$ is a generic DP-ASF and $\bS_+^{2} $ denotes the set of $2\times 2$ PSD matrices. To incorporate this structure in an estimator, we introduce a parametric representation of the DP-ASF.
\begin{figure}
	\centering
	\begin{subfigure}{.5\textwidth}
		\centering
		\includegraphics[width=1\linewidth]{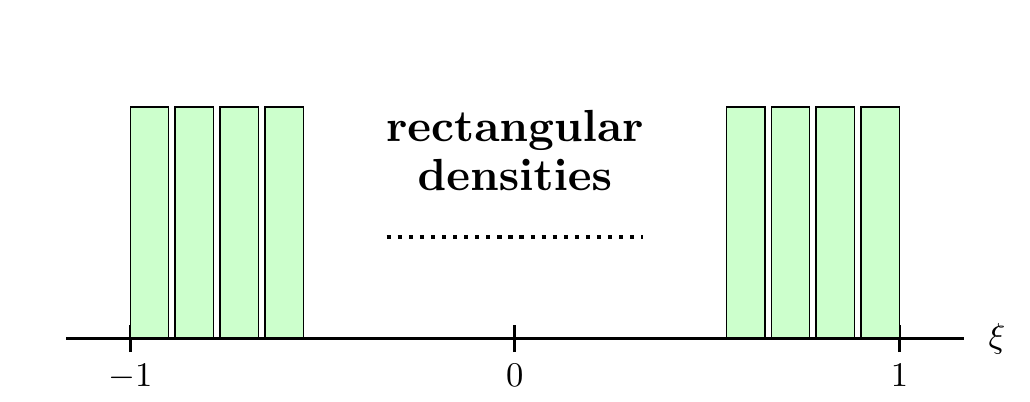}
				\caption{}
		\label{fig:rect_densities}
	\end{subfigure}%
	\begin{subfigure}{.5\textwidth}
		\centering
		\includegraphics[width=1\linewidth]{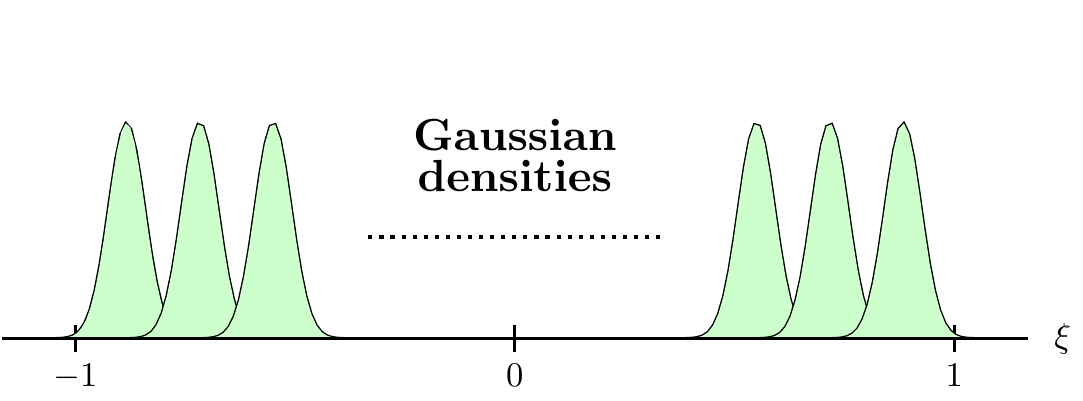}
				\caption{}
				\label{fig:gauss_densities}
	\end{subfigure}
	\caption{\small Examples of density function families: (a) rectangular densities and (b) Gaussian densities.}
	\label{fig:densities}
\end{figure}

\subsection{Parametric Representation of the DP-ASF}
The DP-ASF of a channel models the received power density over each AoA. This power density in turn depends on the scattering properties of the environment: partly it comes from line of sight (LoS) propagation, specular reflection and wedge diffraction in the environment, that occupy narrow angular intervals, while the rest of the power comes from diffuse scattering, occupying wide angular intervals \cite{degli2011analysis} (see Fig. \ref{fig:DPASF}). In order to distinguish between these two types of multipath effects, we decompose the DP-ASF into \textit{discrete} and \textit{continuous} components:
\begin{equation}\label{eq:decompose}
\Gammam (\xi) = \Gammam_d (\xi)+\Gammam_c (\xi),
\end{equation}
where $\Gammam_c (\xi)$ is the continuous component and $\Gammam_d (\xi)$ is the discrete component. For the discrete part, the parametric form is simply given by a train of weighted delta functions:
\begin{equation}\label{eq:discrete_Gamma}
\Gammam_d (\xi)  = \sum_{i=1}^{r} \Cm_i \delta (\xi - \xi_i)
\end{equation}
where $\Cm_i\succeq \mathbf{0},~i=1,\ldots,r$ are $2\times 2$ PSD matrices and $\xi_i,~i=1,\ldots,r$ are discrete AoAs. In contrast, we can not assume a parametric description of $\Gammam_c $ in terms of delta functions. Instead, we define a dictionary of $n$ density functions with small support $\clG_c:=\{\psi_i(\xi)\ge\, 0\, \forall\, \xi \in [-1,1] : i=1,\ldots,n\}$, using which we can approximate $\Gammam_c$ as
\begin{equation}\label{eq:Gamma_c_approx}
\Gammam_c (\xi) \approx  \sum_{i=1}^{n} \Cm_i' \psi_i(\xi), 
\end{equation}
where similar to \eqref{eq:discrete_Gamma} $\Cm_i',~i=1,\ldots,n$ are $2\times 2$ PSD matrices. If $\Gc_c$ is suitably chosen and is large enough ($n\gg 1$), then one can find the coefficients $\Cm_i'$ such that the approximation error in \eqref{eq:Gamma_c_approx} is negligible. Fig. illustrates the approximation of the continuous part of the ASF corresponding to the horizontal channel $[\Gammam (\xi)]_{1,1}=\gamma_{\sf H} (\xi)$ by rectangular densities.

Using \eqref{eq:ch_cov}, \eqref{eq:discrete_Gamma} and \eqref{eq:Gamma_c_approx}, we can derive a similar discrete-continuous decomposition for the channel covariance as
\begin{equation}\label{eq:cov_decomp_pol}
\begin{aligned}
\Sigmah^{ul} =  \Sigmah^{ul,d} + \Sigmah^{ul,c}  &=  \sum_{i=1}^{r}   \Cm_i \otimes  \Am_{ul} (\xi_i) + \int_{-1}^1 \Gamma_c (\xi)\otimes  \Am_{ul} (\xi) d \xi\\
& \approx \sum_{i=1}^{r}   \Cm_i \otimes  \Am_{ul} (\xi_i) + \sum_{i=1}^{n}   \Cm'_i \otimes  \Am'_{ul, i}, 
\end{aligned}
\end{equation}
where we have defined $\Am'_{ul,i} = \int_{-1}^1 \psi_i(\xi)  \Am_{ul} (\xi) d\xi ~\in\bC^{M\times M}$. If the discrete AoAs $\{ \xi_i \}_{i=1}^r$ were known, we could claim via Eq. \eqref{eq:cov_decomp_pol} that estimating $\Sigmam_{\hv}^{ul}$ is equivalent to estimating the coefficient matrices $\{ \Cm_i \}_{i=1}^r$ and $\{ \Cm_i' \}_{i=1}^n$. In order to make this strategy plausible, we first propose a method for estimating the discrete AoAs $\{ \xi_i \}_{i=1}^r$.

\begin{figure}[t]
	\centering
	\includegraphics[ width=0.5\textwidth]{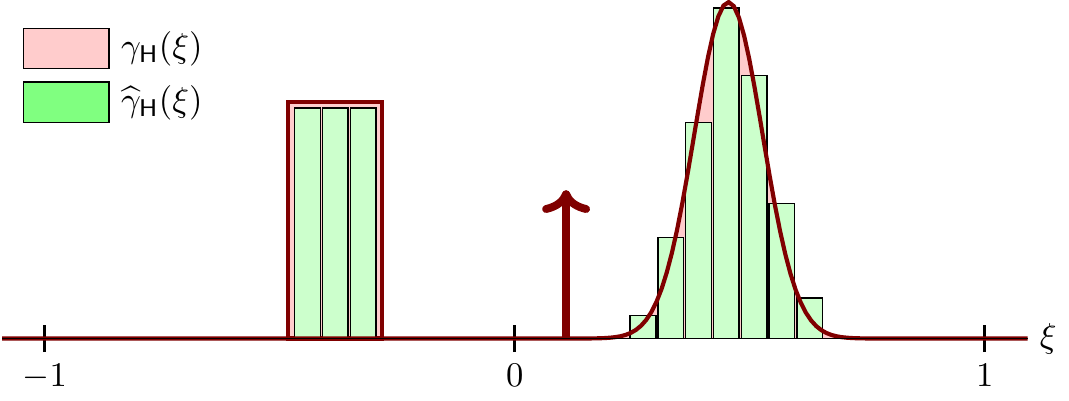}
	\caption{\small An example of approximating the continuous horizontal ASF with a dictionary of rectangular densities.}
	\label{fig:ASF_approx}
\end{figure} 
\subsection{Estimating Discrete AoAs} 
We propose a heuristic method for estimating discrete AoAs, based on the Multiple Signal Classification (MUSIC) algorithm, which is a well-known spectral estimation method \cite{stoica1989music}. Suppose we have an estimate of the number of discrete AoAs as $\widehat{r}$. This implies that the discrete covariance component $\Sigmah^{ul,d} = \sum_{i=1}^{\widehat{r}}   \Cm_i \otimes  \Am_{ul} (\xi_i) $ is of maximum rank $2\widehat{r}$. Define the eigen-decomposition of $\Sigmah^{ul}$ as $\Sigmah^{ul} = \Um \Dm \Um^\herm,$
where $\Um \in \bC^{2M\times 2M}$ is a unitary matrix and $\Dm \in \bR_+^{2M\times 2M}$ is diagonal with real, non-negative elements. We call the space spanned by the set of $2\widehat{r}$ dominant eigenvectors of $\Sigmah^{ul}$ by ``signal subspace", and the space spanned by the remaining $2M-2\widehat{r}$ eigenvectors as ``noise subspace". Assume the diagonal elements of $\Dm$ to be ordered as $d_1\ge d_2\ge \ldots\ge d_{2M}$ and denote their associated eigenvectors as $\uv_1,\uv_2,\ldots,\uv_{2M}$. The vectors spanning the noise subspace are collected in the matrix $\bfU_\text{noi}=[\uv_{2\widehat{r}+1}, \dots, \uv_{2M}]$. Then we form the \textit{pseudo-spectrum function} $\eta (\xi) = \left\Vert  \bfU_\text{noi}^\herm\, \left(\mathbf{I}_2 \otimes \av (\xi) \right) \right\Vert_{\sf F}^2$ and estimate the discrete AoAs as the $\widehat{r}$ minimizers of $\eta (\xi) $ with the smallest minimum value. Intuitively, in this way we find a number of $\widehat{r}$ AoAs that the $2M\times 2$ dual-polarized array response $\mathbf{I}_2 \otimes \av (\xi)$ along them, has the smallest norm when projected to the noise subspace. This heuristic follows the same rationale as the classical MUSIC method, in which the $M$-dimensional array response vector of an $M$-dimensional ULA is projected to the noise subspace and measured in terms of the $\ell_2$-norm to compute the pseudo-spectrum function. After finding the minima of $\eta (\xi)$, we identify its $\widehat{r}$ smallest minima as the estimated discrete AoAs and denote them as $\widehat{\xi}_i,\, i=1,\ldots,\widehat{r}$. Recalling \eqref{eq:cov_decomp_pol}, now we can say that estimating $\Sigmam_{\hv}^{ul}$ is equivalent to estimating the $n+\wh{r}$ coefficient parameters, namely $\{ \Cm_i \}_{i=1}^{\widehat{r}}$ and $\{ \Cm_i' \}_{i=1}^n$.

\begin{remark}
The number of spikes is typically a few and can be learned over time. Also, overestimating the number of spikes is better than underestimating it: if \textit{fake} spikes (i.e., false positives) appear in the set of estimated discrete AoAs, they will be eventually associated with small coefficients in the next coefficient estimation step. However, if a true spike is not detected, then we may not get an accurate covariance estimate as no term in the parametric expansion \eqref{eq:cov_decomp_pol} will compensate for the contribution of the missing spike. Therefore, we use a large-enough pre-defined value for $\widehat{r}$.\hfill  $\lozenge$
\end{remark}
\subsection{Estimating DP-ASF Coefficients}
Let us first reformulate the channel covariance parametric description in a simpler form. Define the known $M\times M$ matrices $\Sm_i = \Am_{ul} (\widehat{\xi}_i)$ for $i=1,\ldots,\widehat{r}$ and $\Sm_i = \Am'_{ul,i-\widehat{r}} $ for $i=\widehat{r}+1,\ldots,\widehat{r}+n$. Also define their associated unknown coefficients as $\Wm_i = \Cm_i$ for $i=1,\ldots,\widehat{r}$ and $\Wm_{\wh{r}+i} = \Cm'_{i} $ for $i=1,\ldots,n$. Then \eqref{eq:cov_decomp_pol} can be reformulated as 
\begin{equation}\label{eq:cov_parametric}
\Sigmam_{\hv}^{ul} (\{ \Wm_i \}_{i=1}^{\wh{r}+n})\approx \sum_{i=1}^{\widehat{r}+n} \Wm_i \otimes \Sm_i.
\end{equation}
Now, the problem is to estimate the coefficient matrices $\{\Wm_i \in \bS_+^2\}_{i=1}^{n+\wh{r}}$, given noisy pilot measurements $\{\yv_{ul} (j)\}_{j=1}^N$ in \eqref{eq:noisy_pilots}. Our proposition for performing this task is based on minimizing the the difference between the channel sample covariance matrix and its parametric form as a function of the coefficients. We perform the minimization by constraining the coefficients to be PSD. 
Formally, we have the following optimization problem:
\begin{equation}\label{eq:PSDLS}
\begin{aligned}
\{ \widehat{\Wm}_{i} \}_{i=1}^{\widehat{r}+n}\, =\, & \underset{\{ \Wm_i \}_{i=1}^{\widehat{r}+n}}{\arg\min} && \Vert \widehat{\Sigmam}_{\hv}^{ul} - \sum_{i=1}^{\widehat{r}+n} \Wm_i\otimes \Sm_i \Vert_{\sf F}^2 \\
& \text{subject to} && \Wm_i\succeq \mathbf{0},~i=1,\ldots,\widehat{r}+n.
\end{aligned} 
\end{equation}
We call this problem a \textit{positive semi-definite least-squares} (PSD-LS) program. The PSD-LS is convex and can be solved using standard algorithms (SDP solvers). Then we obtain the covariance estimate simply by using \eqref{eq:cov_parametric} and replacing $\Wm_i$ with $\wh{\Wm}_{i}$. Note that solving \eqref{eq:PSDLS} also provides an estimate of the DP-ASF using \eqref{eq:discrete_Gamma} and \eqref{eq:Gamma_c_approx} as
\begin{equation}\label{eq:DPASF_est}
\widehat{\Gammam}(\xi) = \sum_{i=1}^{\wh{r}} \wh{\Wm}_{i}\, \delta (\xi -\wh{\xi}_i) + \sum_{i=1}^{n} \wh{\Wm}_{\wh{r}+i}\,  \psi_i(\xi).
\end{equation}
\begin{figure*}[t]
	\centering
	\includegraphics[scale=0.6]{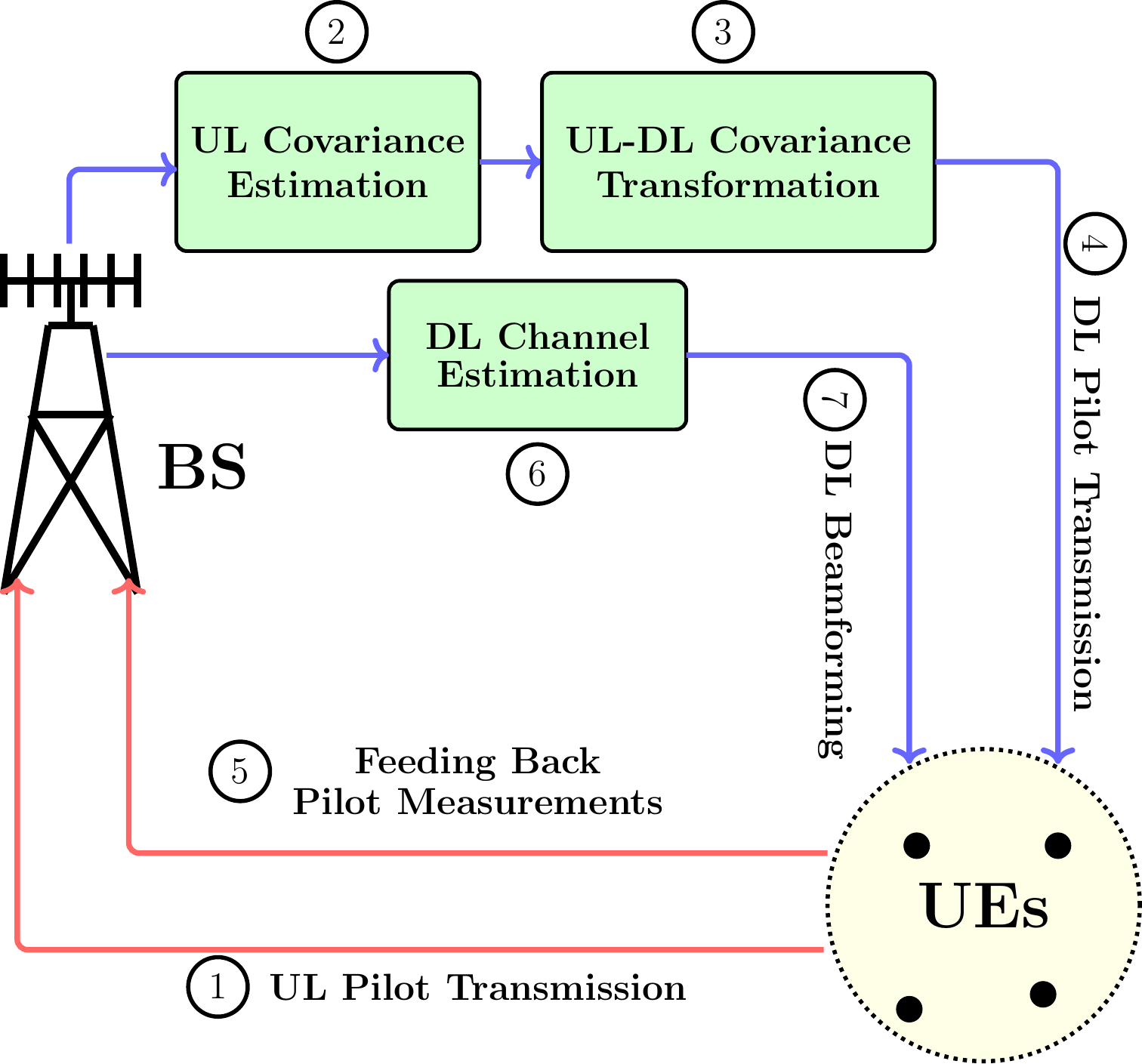}
	\caption{Overall diagram of our scheme.}
	\label{fig:block_diagram}
\end{figure*} 
\section{UL-DL Covariance Transformation}\label{sec:transformation}
Estimating DL channel covariance is necessary for MMSE DL channel estimation and multi-user common DL channel training. Once UL channel covariance is estimated, estimating the DL channel covariance in TDD mode is straightforward, as due to channel reciprocity, UL and DL covariances are identical. However, in FDD mode the covariance varies from UL to DL due to a change of frequency band, resulting in a different response by array elements. 

A useful property of estimating the ASF in parametric form, as we did in the previous section, is that using it we can estimate the DL covariance with a ``change of basis". Similar to the UL channels, the H and V channels in the DL can be represented as
\begin{equation}\label{eq:h_polar_dl}
\begin{aligned}
\hv_{H}^{dl} = \int_{-1}^1 W_{ H} (\xi) \bv (\xi) d\xi, \quad 
\hv_{V}^{dl} = \int_{-1}^1 W_{V} (\xi) \bv (\xi) d\xi,
\end{aligned}
\end{equation}
where $\bv (\xi)$ is the DL array response vector. Assuming as before the antenna spacing $d=\frac{\lambda_{ul}}{2}$ we have
$ \bv (\xi) = [1,e^{j\pi \nu \xi},\ldots,e^{j\pi (M-1)\nu \xi}]^\transp,$ where $\nu=\frac{\lambda_{ul}}{\lambda_{dl}}=\frac{f_{dl}}{f_{ul}} $ is the DL to UL carrier frequency ratio. With the same reasoning leading to \eqref{eq:ch_cov}, we express the DL covariance as
\begin{equation}\label{eq:dl_cov}
\Sigmam_{\hv}^{dl} = \int_{-1}^{1} \Gammam (\xi) \otimes \Am_{dl} (\xi) \, d\xi,
\end{equation} 
where $\Am_{dl} (\xi) = \bv (\xi) \bv(\xi)^\herm$. From the estimate of $\Gammam (\xi)$ in \eqref{eq:DPASF_est} we estimate $\Sigmam_{\hv}^{dl} $ as
\begin{equation}\label{eq:dl_cov_estimate}
\widehat{\Sigmam}_{\hv}^{dl} = \int_{-1}^1 \widehat{\Gammam}(\xi)\otimes  \Am_{dl} (\xi)d\xi =  \sum_{i=1}^{n+\widehat{r}}  \widehat{\Wm}_{i} \otimes \Sm_i',
\end{equation}
where $\Sm_i' = \int_{-1}^{1} \psi_i (\xi) \Am_{dl} (\xi)d\xi$ for $i=1,\ldots, n$ and $\Sm_i'=\Am_{dl} (\widehat{\xi}_i) $ for $i=n+1,\ldots,n+\widehat{r}$. 

To summarize, we have so far developed a method for estimating DL channel covariance from UL pilots for every user. The necessity of DL covariance acquisition becomes clear in the next section.
 
\section{Downlink Channel Training and Multi-User Precoding}\label{sec:common_training}
Besides the problem of covariance estimation, the BS is required to transmit multiplexed data to several users in the DL. An interference-free transmission is possible only if the BS has the instantaneous DL channel state information (CSI) for all users to construct a beamformer. Since channel reciprocity does not hold in FDD mode, the instantaneous DL CSI is obtained via common DL training (pilot transmission) of the user channels and feeding back the measurements to the BS during UL. The challenge is that, for a dual-polarized massive MIMO system with a channel dimension of $2M\gg 1$, the number of pilots used for DL training must be large so that channel estimation is feasible. This results in a substantial reduction of DL sum-rate. Also feeding back a large number of measurements to the BS consumes a considerable part of UL resources and may result in large delays.

In order to overcome this dimensionality bottleneck, recently we proposed the \textit{active channel sparsification} (ACS) method \cite{khalilsarai2018fdd}, which enables stable channel estimation for any given pilot dimension that is specified by the standard. In particular, ACS aims at designing a linear precoder that is concatenated with the physical channel. The design of the precoder depends only on user DL covariances, and obviously not on the instantaneous channel realizations as, in fact, they should be estimated via UL closed-loop feedback. This completes our overall proposed scheme for implementing a dual-polarized FDD massive MIMO system, as illustrated in the block-diagram of Fig. \ref{fig:block_diagram}.

We can formalize the idea behind ACS as follows. To jointly train the DL channels, the BS transmits a pilot matrix $\Psim$ of dimension $\Tdl \times M'$, where $\Tdl\le T$ is a fixed pilot dimension such that each row $\Psim_{i,.}$ represents a pilot signal that is transmitted from the $M' \leq 2M$ inputs of a precoding matrix $\Bm$ of dimension  $M' \times 2M$. The integer $M'$ is a suitable intermediate dimension that, as we will see later, is determined during the precoder design. The observed training symbols at user $k$ can be expressed via the $\Tdl$-dimensional vector
	\begin{equation}\label{eq:cs_eq_1}
{\bfy}_{dl,k}= {\bf\Psi} \bfB \hdlk + {\bfz}_{k} = {\bf\Psi} \widetilde{\hv}_{dl,k} + {\bfz}_k,
\end{equation}
where $\hv_{dl,k}$ is the DL channel vector of user $k$ for $k=1,\ldots,K$, ${\bfz}_k\sim \cg (\mathbf{0},N_0 \mathbf{I}_{\Tdl} )$ is the AWGN, and pilot and precoding matrices are normalized such that $\trace ( \Psim \Bm \Bm^\herm \Psim^\herm ) = \Tdl \Pdl, $ where $\Pdl$ is the BS transmit power resulting in the DL signal-to-noise ratio (SNR) to be equal to $\text{SNR} =\frac{\Pdl}{N_0}$. 

In \eqref{eq:cs_eq_1} we have also defined the \textit{effective channel vector} $\widetilde{\hv}_{dl,k} := \Bm \hv_{dl,k}$ as the concatenation of the precoder with the true channel. In the ACS method, our intention is to design $\Bm$ as a {\em sparsifying} precoder, such that each user effective channel vector $\widetilde{\hv}_{dl,k} $ is sufficiently sparse (over the angular domain) and yet the collection of the effective channels for $k = 1,\ldots,K$ forms an effective channel matrix with a rank that is as large as possible. In this way, each effective channel can be estimated using the fixed (possibly even small) pilot overhead $\Tdl$, but the BS is still able to transmit multiple data streams in the DL.

\subsection{Necessity of Channel
Sparsification}\label{sec:nec_acs} 
The channel vector of user $k$ admits the Karhunen-Lo{\`e}ve (KL) expansion $\hv_{dl,k} = \sum_{m=1}^{2M} g_{k,m}\allowbreak \sqrt{\lambda_{k,m}} \,\uv_m^{(k)}$, where $g_{k,m}\sim \cg (0,1) $ are i.i.d. complex Gaussian variables, $\uv_m^{(k)}$ is the $m$-th eigenvector of user $k$ DL channel covariance and $\lambda_{k,m}$ is its associated eigenvalue. Define the vector of eigenvalues of user $k$ as $\lambdav_k\in \bR_+^{2M}$ and define the support of $\lambdav_k$ as $\Sc_k = \{ m : \lambda_{k,m} \neq 0 \}$ with a size $s_k = |\Sc_k|$, which specifies the covariance rank. The following lemma yields necessary and sufficient conditions for the stable estimation of $\hv_{dl,k}$, where by estimation stability we mean that the estimation error vanishes as the noise variance tends to zero. 
\begin{lemma}\label{lem:stable_rec}
Consider the sparse Gaussian vector $\hdlk$ with support set $\Sc_k$. 
Let $\widehat{\hv}_{dl,k}$ denote any estimator for $\hdlk$ 
based on the observation
$\yv_{dl,k} = {\bf\Psi} \hdlk + {\bfz}_k$ (note that this coincides with 
(\ref{eq:cs_eq_1}) by replacing $\Bm=\mathbf{I}_{2M}$, i.e., without the sparsifying precoder). Let $\Rm_e = \bE[ (\hdlk - \widehat{\hv}_{dl,k}) (\hdlk - \widehat{\hv}_{dl,k})^\herm ]$ denote the corresponding estimation error covariance matrix.  
If  $\Tdl \ge s_k$ there exist pilot matrices $\Psim \in \bC^{\Tdl \times 2M}$ for which $\lim_{\Nvar \downarrow 0} \trace(\Rm_e) = 0$ for all support 
sets $\Sc_k : |\Sc_k| = s_k$. 
Conversely,  for any support set $\Sc_k : |\Sc_k| = s_k$ any pilot matrix $\Psim \in \bC^{\Tdl \times 2M}$ 
with  $\Tdl < s_k$ yields $\lim_{\Nvar \downarrow 0} \trace(\Rm_e) > 0$.
\hfill $\square$
\end{lemma}
\begin{proof}
See the proof of Lemma 1 in \cite{khalilsarai2018fdd}.
\end{proof}

Lemma \ref{lem:stable_rec} asserts the following important implication. First, note that stable channel estimation is necessary in order to achieve high spectral efficiency in the high-SNR regime. In fact, if the estimation mean-squared error (MSE) of the user channels does not vanish as $\Nvar \downarrow 0$, the system self-interference due to imperfect channel knowledge grows proportionally to the signal power and we have an interference-limited multi-user system, which is undesirable. On the other hand, if $\Tdl < s_k$ for some user $k$, then  any scheme that relies on channel sparsity will fail to yield a stable channel estimate. This includes, for example, the sophisticated compressed sensing (CS) methods, which simply can not stably estimate a $s_k$-sparse channel from $\Tdl < s_k$ measurements. Therefore, one constraint for designing the sparsifying precoder $\Bm$, is that once it is applied to the channel vector, the sparsity of the resulting effective channel is less than or equal to the available pilot dimension $\Tdl$. 

\subsection{Virtual Beam Representation}
From the discussion above, it seems to be necessary that all the channel vectors have a sparse representation over a \textit{common} dictionary. The reason is that, otherwise each channel has its sparse representation over an entirely different dictionary than another channel and it becomes extremely difficult to design a precoder that jointly sparsifies all channels. We call the atoms of the common dictionary as ``virtual beams". We want the virtual beams to be (at least approximately) equivalent to a set of eigenvectors, shared among all user channel covariances. This ensures that the number of beams that significantly contribute to the channel is not very different from the channel sparsity, as reflected in the KL expansion of each user channel. For covariances of dimension $2M$, this translates to finding a unitary matrix $\Vm,~\Vm^\herm \Vm = \mathbf{I}_{2M}$ that (approximately) diagonalizes all user channel covariances, i.e. $
\Vm^\herm \Sigmam_k^{dl} \Vm \approx \text{diag}(\dv_k),$  for $k=1,\ldots,K$ where $\dv_k$ is a $2M$-dimensional non-negative vector and the approximation is understood in the sense that a distance measure between the LHS and the RHS is sufficiently low. Fortunately, for a dual-polarized ULA such an approximate common eigenvector set exists. First, note that we can express a generic dual-polarized ULA covariance $\Sigmam$ in four blocks as
\begin{equation}\label{eq:dp_cov_block}
\Sigmam = \begin{bmatrix}
\Sigmam_{\sf HH} & \Sigmam_{\sf HV} \\
\Sigmam_{\sf VH} & \Sigmam_{\sf VV} 
\end{bmatrix}, 
\end{equation}
where $\Sigmam_{\sf HH} = \bE [ \hv_{\sf H}\hv_{\sf H}^\herm ]$, $\Sigmam_{\sf VV} = \bE [ \hv_{\sf V}\hv_{\sf V}^\herm ]$, and $\Sigmam_{\sf HV} = \Sigmam_{\sf VH}^\herm  = \bE [ \hv_{\sf H}\hv_{\sf V}^\herm ]$, where $\hv_{H}$ and $\hv_{V}$ are generic H and V channel vectors. The diagonal blocks $\Sigmam_{\sf HH} $ and $\Sigmam_{\sf VV} $ are Hermitian Toeplitz matrices of dimension $M$. The well-known Szeg{\"o} theorem states that for a Hermitian Toeplitz matrix of dimension $M\gg 1$, there exists a circulant matrix that approximately has the same eigenvalue distribution as the Toeplitz matrix \cite{adhikary2013joint}. Let us denote the circulant approximation of a generic Toeplitz matrix $\Tm$ by $\mathring{\Tm}$. The eigenvectors of a Hermitian circulant matrix are given by the DFT columns of the same size, namely by the columns of a matrix $\Fm\in \bC^{M\times M}$ where $[\Fm]_{m,n}=\frac{1}{\sqrt{M}}e^{j2\pi \frac{(m-1)(n-1)}{M}},~m,n=1,2,\ldots,M$. Therefore, we have $\mathring{\Tm}= \Fm \text{diag} (\mathring{\lambdav}) \Fm^\herm$, for some $\mathring{\lambdav}\in \bR^M$. From the Szeg{\"o} theorem it follows that the DFT matrix approximately diagonalizes large Toeplitz matrices. Applied to the problem in hand, we can compute the circulant approximation for $\Sigmam_{\sf HH}$ and $\Sigmam_{\sf VV}$ in a constructive way by defining
\begin{equation}\label{eq:vars}
[\mathring{\lambdav}_{\sf H}]_m=[\Fm^\herm \Sigmam_{\sf HH} \Fm ]_{m,m}, ~~ [\mathring{\lambdav}_{\sf V}]_m=[\Fm^\herm \Sigmam_{\sf VV} \Fm ]_{m,m}
\end{equation}
 and setting $\mathring{\Sigmam}_{\sf HH}=\Fm\, \text{diag} \left(\mathring{\lambdav}_{\sf H}\right) \Fm^\herm$ and $\mathring{\Sigmam}_{\sf VV} = \Fm\, \text{diag} \left(\mathring{\lambdav}_{\sf V}\right) \Fm^\herm$. Then, from the Szeg{\"o} theorem we have 
$\Sigmam_{\sf HH}\approx\Fm\, \text{diag} \left(\mathring{\lambdav}_{\sf H}\right) \Fm^\herm,\quad \Sigmam_{\sf VV}\approx\Fm\, \text{diag} \left(\mathring{\lambdav}_{\sf V}\right) \Fm^\herm$.
It follows that the ${\sf H}$ and ${\sf V}$ channel vectors admit a (approximate) representation over the columns of $\Fm=\left[ \fv_0,\ldots,\fv_{M-1} \right]$ as
$ \hv_{\sf H} \approx \Fm  \gv_{\sf H},~\hv_{\sf V} \approx \Fm  \gv_{\sf V}$,
where $\gv_{\sf H}\sim \cg \left(\mathbf{0},\text{diag} \left( \mathring{\lambdav}_{\sf H} \right) \right)$ and $\gv_{\sf V}\sim \cg \left(\mathbf{0},\text{diag} \left( \mathring{\lambdav}_{\sf V} \right) \right)$ are i.i.d complex Gaussian random vectors. The elements $[\mathring{\lambdav}_{\sf H}]_m$ and $[\mathring{\lambdav}_{\sf V}]_m$ are an approximation of the variance of the $H$ and $V$  channel coefficients along the $m$-th virtual beam $\fv_m$. We call the $2M$-dim vector $\gv = [\gv_{\sf H}^{\transp},\gv_{\sf V}^{\transp}]^\transp$ \textit{the dual-polarized channel coefficients vector}.  

From the discussion above we conclude that the dual-polarized DL channel vector of the $k$-th user $\hdlk$ is related to its corresponding channel coefficients as
\begin{equation}\label{eq:h_dl_represent}
\hv_{dl,k} \approx \widetilde{\Fm} \gv_k,\, \,k=1,\ldots,K
\end{equation}
where the Kronecker product $\widetilde{\Fm}  = \mathbf{I}_2 \otimes \Fm $ represents the set of common virtual beams for the dual-polarized channel among all users. For every $m$, the elements $[\gv_{\sf H,k}]_m$ and $[\gv_{\sf V,k}]_m$ are correlated, due to the correlation between horizontal and vertical channels. Note that, representing the channel coefficients over the angular domain, $\gv_k$ is usually a sparse vector in the massive MIMO regime, i.e. it has significantly large elements only over a limited set of indices $\Jc_k$, known as the support set such that $|\Jc_k|\ll 2M$. 
\subsection{User-Virtual Beam Graph Representation}
Let us define the channel matrix as $\Hm = [\hv_{dl,1},\ldots,\hv_{dl,K}]\in \bC^{2M\times K}$, which is related to the matrix of channel coefficients $\Gm = [\gv_1,\ldots,\gv_K]\in \bC^{2M\times K}$ as $\Hm = \widetilde{\Fm}\Gm$. 
\begin{remark}\label{remark:G_dependence}
	The elements of the coefficients matrix $\Gm$ are not  i.i.d, but they entail a special type of dependence: any Gaussian element $[\Gm]_{m,k}$ $k=1,\ldots,K,~m=1,\ldots,M$, is correlated with (at most) a single element $[\Gm]_{M+m,k}$, namely its peer coefficient for the vertical polarization. 
	\end{remark}
Since $\widetilde{\Fm}$ is a unitary matrix, we have
$\rank (\Hm) = \rank (\Gm)$, which is a useful identity, since now claims about the rank of $\Gm$ immediately carry over to those about the rank of $\Hm$. The following lemmas relate the rank of $\Gm$ to a graph-theoretic property.  

\begin{lemma}\label{lem:CUR}
	[Skeleton decomposition \cite{goreinov1997theory}] Consider $\Gm\in \bC^{2M\times K}$, of rank $r$. Let $\Qm$ be an $r\times r$ non-singular
	intersection submatrix obtained by selecting $r$ rows and $ r$ columns of $\Gm$. Then, we have $\Gm = \Lm \Om \Rm$,
	where $\Lm \in \bC^{2M\times r}$ and $\Rm \in \bC^{r\times K}$ are the matrices of the selected columns and rows forming the intersection $\Qm$ and $\Om=\Qm^{-1}$.
	\hfill $\square$
\end{lemma}
\begin{lemma}\label{lem:matching}
	[Rank and perfect matchings] Let $\Qm$ denote an $r\times r$ matrix with some elements identically zero, and the non-identically zero elements drawn from a continuous distribution, such that an element $[\Qm]_{i,j}$ is independent from all elements that are not in the same row or column with it (it may or may not be dependent on elements in the same row or same column). Consider the associated
	bipartite graph with adjacency matrix $\Am$ such that $\Am_{i,j}=1$ if $\Qm_{i,j}$ is not identically zero, and $\Am_{i,j} = 0$ otherwise.
	Then, $\Qm$ has rank $r$ with probability 1 if and only if the associated bipartite graph contains a perfect matching. \hfill $\square$
\end{lemma}
\begin{proof}
	The determinant of $\Qm$ is given by the expansion $\text{det}(\Qm)=\sum_{\iota\in \boldsymbol{\pi}_r} \text{sgn}(\iota) \prod_{i} [\Qm]_{i,\iota (i)}$, where $\iota$ is a permutation of the set $\{1,2,\ldots,r\}$, where $\boldsymbol{\pi}_r$ is the set of all such permutations and where $\text{sgn}(\iota)$ is either 1 or -1. The product $\prod_{i} [\Qm]_{i,\iota (i)}$ is non-zero only for the perfect matchings in the bipartite graph. Hence, if the bipartite graph contains a perfect matching, then $\text{det}(\Qm)\neq 0$ with probability 1 (and $\text{rank}(\Qm)=r$), since the non-identically zero entries of $\Qm$ are drawn from a continuous distribution, such that all elements involved in the product $\prod_{i} [\Qm]_{i,\iota (i)}$ are independent (no two elements from either the same row or the same column are involved in this product). If it does not contain a perfect matching, then $\text{det}(\Qm)=0$ and therefore $\text{rank}(\Qm)<r$. \hfill 
\end{proof}	
The following corollary emerges from a combination of Lemmas \ref{lem:CUR} and \ref{lem:matching}.
\begin{corollary}\label{corollary:matching}
	The rank $r$ of the random channel coefficients matrix $\Gm \in \bC^{2M\times K}$, with the particular statistical structure explained in Remark \ref{remark:G_dependence} is given, with probability 1, by the size of the largest intersection submatrix whose associated bipartite graph (defined as in Lemma \ref{lem:matching}) contains a perfect matching.   \hfill $\square$
\end{corollary}
This corollary implies that we can study the rank properties of $\Gm$ by associating to it a bipartite graph. On one side of this graph we have $2M$ nodes, representing the $2M$ virtual beams (columns of $\widetilde{\Fm}$) and on its other side, we have $K$ nodes representing the users. The nodes are connected according to the adjacency matrix as introduced in Lemma \ref{lem:matching}. From a practical viewpoint, a user node is connected to a virtual beam node, if and only if the the user channel has a ``strong enough" coefficient along that beam (this point becomes clear shortly). Then we know from Corollary 1 that maximizing rank of $\Gm$ (hence rank of $\Hm$) is equivalent to maximizing the matching size in a sub-graph of the user-virtual beam bipartite graph. The sub-graph corresponds to those beams and users that will be eventually present in the effective channel matrix. This sub-graph can not be selected arbitrarily, but such that the number of significant channel coefficients (channel sparsity over dictionary $\widetilde{\Fm}$) for any user in the effective channel matrix induced by the sub-graph must be less than the pilot dimension $\Tdl$ so that stable channel estimation is possible according to Lemma \ref{lem:stable_rec}.

Let us introduce the user-virtual beam bipartite graph as $\Gc (\Vc,\Kc,\Ec)$, where $\Kc$ denotes a set of $K$ nodes on one side of the graph representing the users and $\Vc$ is a set of $2M$ nodes representing the virtual beams (columns of $\widetilde{\Fm}$). The node $k\in \Kc$ is connected to a virtual beam $m\in \Vc$ if and only if the variance of the channel coefficient of user $k$ along virtual beam $m$ is greater than a predefined threshold $\varepsilon>0$, i.e. $\Ec = \{  (k,m)\,:\, \bE [ |[\gv_k]_m|^2 ] \ge \varepsilon \}$.
Since $\widetilde{\Fm}$ is a block diagonal matrix with the DFT matrix $\Fm$ as its diagonal blocks, the variance of user $k$ along the first $M$ nodes of $\Vc$ is given by the vector of horizontal channel coefficient variances $\lambdav_{{\sf H}, k}$ and its variance along the second $M$ nodes of $\Vc$ is given by the vector of vertical channel coefficient variances $\lambdav_{{\sf V},k}$ (see \eqref{eq:vars}). Define the $2M$-dim vector of coefficient variances for user $k$ as $\lambdav_k = [\lambdav_{{\sf H},k}^{\transp},\lambdav_{{\sf V},k}^{\transp}]^\transp$. Then an edge between nodes $m\in \Vc$ and $k \in \Kc$ exists if and only if $ [\lambdav_k]_m \ge \varepsilon$ and the weight assigned to this edge is defined as $[\lambdav_k]_m$. These conventions define the adjacency matrix $\Am \in \bR^{2M\times K}$ and its weighted version $\Wm$ where $[\Wm]_{m,k}=[\lambdav_k]_m$ for $[\lambdav_k]_m\ge\varepsilon$ and $[\Wm]_{m,k}=0$ otherwise. See Fig. \ref{fig:bigraph} for an example of the user-virtual beam bipartite graph.

\subsection{Active Channel Sparsification}
Introducing the bipartite graph, we are in a position to transform the effective channel matrix rank maximization problem to the maximum cardinality matching problem over a bipartite graph. This shall be subject to a constraint on the number of significantly large elements (i.e. the channel \textit{sparsity}) in the coefficient vectors $\gv_k,\,k=1,\ldots,K$. Let $\Gc= (\Vc,\Kc,\Ec)$ denote the user-virtual beam bipartite graph as previously defined.
Also let $\Mc (\Vc',\Kc')$ denote a matching of the subgraph $\Gc' =(\Vc',\Kc',\Ec')$ 
of $\Gc$. A matching is a set of edges in a graph, such that no two edges share a vertex \cite{diestel2005graph}. Suppose $\Tdl$ to be the available DL pilot dimension. Maximizing the effective channel rank constrained to the limitation of the effective channel sparsity to $\Tdl$ is equivalent to the following optimization problem:	
\begin{subequations}\label{eq:my_opt_1}
	\begin{align}
	& \underset{\Vc'\subseteq \Vc, \Kc' \subseteq \Kc}{\text{maximize}} && \left\vert \Mc\left(\Vc',\Kc'\right)\right\vert    &&&~ \label{eq:my_opt_1-one}  \\
	& \text{subject to} &&  \text{deg}_{\Gc'} (k) \le \Tdl &&&\forall ~k\in \Kc',  \label{eq:my_opt_1-two} \\
	& ~ &&  \sum_{m\in \Nc_{\Gc'} (k) } [\Wm]_{m,k}  \ge P_0, &&& \forall ~ k\in \Kc', \label{eq:my_opt_1-three}
	\end{align}
\end{subequations} 
where $\text{deg}_{\Gc'} (k)$ denotes the degree of node $k$ in subgraph $\Gc'$, $\Nc_{\Gc'} (k) $ denotes the set of neighbor nodes to $k$ in $\Gc'$. $P_0\ge 0$ is a predefined power threshold. Constraint \eqref{eq:my_opt_1-two} ensures that the number of virtual beams contributing to the channel of user $k$ is less than or equal the pilot dimension $\Tdl$. The number of contributing virtual beams determines the channel sparsity in the beam domain and this constraint satisfies the condition necessary for stable channel estimation (see Lemma \ref{lem:stable_rec}).  Constraint \eqref{eq:my_opt_1-three} is a power constraint, which ensures that if a user is chosen to be served (i.e., is in the solution subgraph), then it should have sufficient power (at least $P_0$) along those virtual beams that contribute to it and are present in the solution subgraph. 
\begin{figure}[t]
	\centering
	\includegraphics[ width=0.45\textwidth]{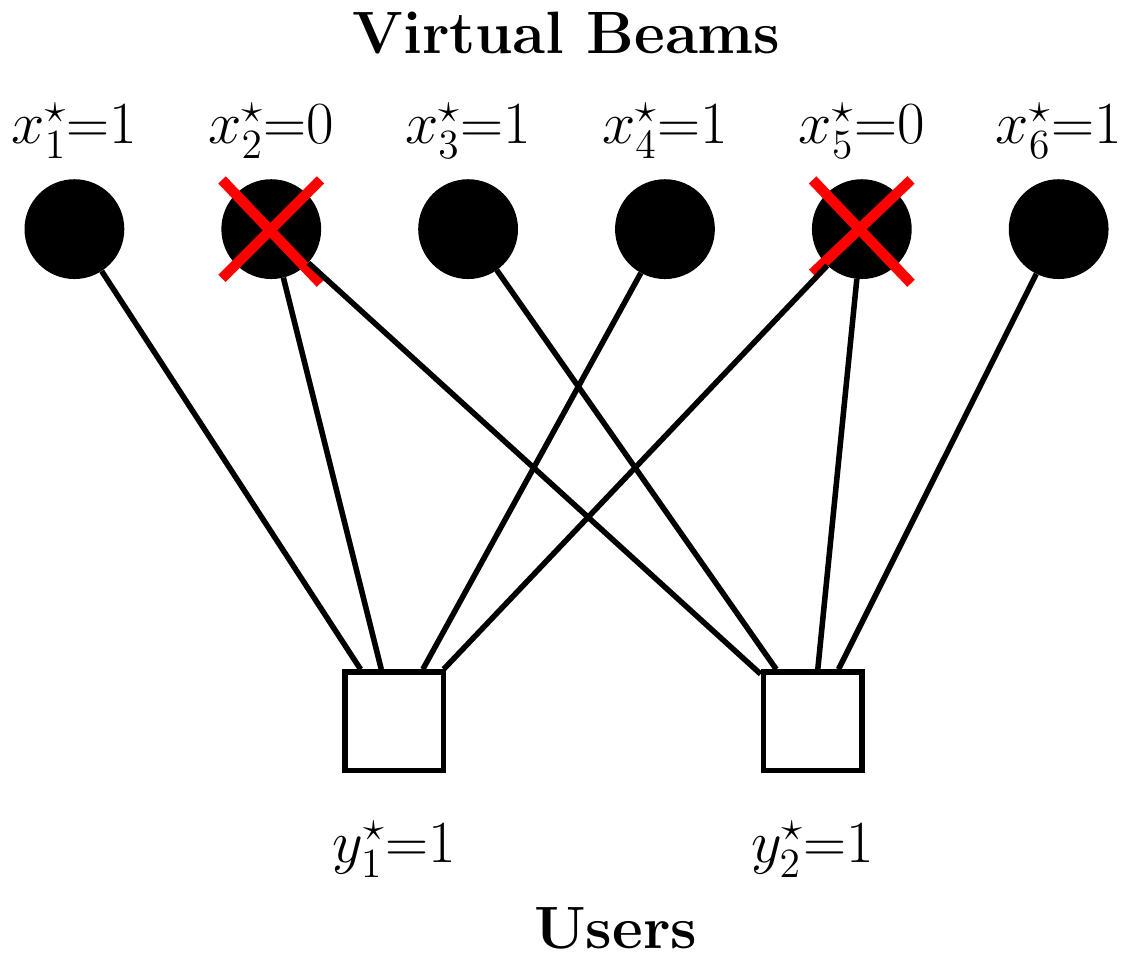}
	\caption{\small an example of a user-virtual beam bipartite graph with $K=2$ users and $2M=6$ virtual beams. The red crosses denote inactive (i.e., eliminated) beams after solving the MILP with $\Tdl=2$. }
	\label{fig:bigraph}
\end{figure} 
\begin{theorem}\label{thm:MILP_formulation}
	An optimal solution to the optimization problem in \eqref{eq:my_opt_1} is given by solving the mixed integer linear program (MILP) below:
	\begin{subequations}\label{opt:P_MILP_Thm}
		\begin{align}
		& \underset{x_m,y_k,z_{m,k} }{\text{maximize}}  && \sum_{m\in \Vc} \sum_{k \in \Kc} z_{m,k} + \delta \sum_{m \in \Ac} x_m  \label{eq:obj_1_thm} \\
		& \text{subject to}  && z_{m,k} \le [\Am]_{m,k} ~~\forall m\in\Vc,k\in \Kc, \label{eq:first_thm}\\
		&~ && \sum_{k\in \Kc}z_{m,k} \le x_m ~~ \forall m\in \Vc,\\ 
		&~ && \sum_{m\in \Vc}z_{m,k} \le y_k ~~ \forall k\in \Kc,  \\
		&~ &&\underset{m\in \Vc}{\sum} [\Am]_{m,k} x_m \le \Tdl y_k + 2M (1-y_k)~\forall k\in \Kc\\
		&~ &&P_0\, y_k \le \sum_{m \in \Vc} [\Wm]_{m,k} x_m ~~\forall k \in \Kc,  \\
		&~&& x_m \le \sum_{k\in \Kc} [\Am]_{m,k} y_k ~~ \forall m\in \Ac, \\
		&~&& x_m, y_k \in \{0,1\} ~~\forall m\in \Vc,k\in \Kc,\\
		&~&& z_{m,k} \in [0,1] ~~\forall m\in \Vc,k\in \Kc \label{eq:last_thm},
		\end{align}
	\end{subequations}
	where $0<\delta<\frac{1}{2M}$ is a small positive scalar. The binary variables $\{x_m \}_{m=1}^{2M}$ represent the virtual beams and the binary variables $\{y_k \}_{k=1}^K$ represent the users. The solution sub-graph is given by the set of nodes $\Vc^\star = \{m:x_m^\star= 1\}$ and $\Kc^\star = \{k:y_k^\star=1\}$, with $\{x_m^\star\}_{m=1}^{2M}$ and $\{y_k^\star\}_{k=1}^K$ being a solution of \eqref{opt:P_MILP_Thm}. \hfill $\square$
\end{theorem}
\begin{proof}
	See the proof of  Theorem 1 in \cite{khalilsarai2018fdd}.
\end{proof}
The MILP introduced in \eqref{opt:P_MILP_Thm} can be solved for most practical array dimensions (for example, up to $M=128$) using standard solvers. We have used the built-in {\sf ``intlinprog"} routine in MATLAB to perform our simulations, provided in Section \ref{sec:simulations}. The solution of \eqref{opt:P_MILP_Thm} determines the set of users as well as virtual beams that are to be probed and served: a user $k$ is probed and served if and only if $y_k^\star=1$; similarly, a virtual beam $m$ is probed and served if and only if $x_m^\star=1$. Fig. \ref{fig:bigraph} provides a miniature example, in which we have $K=2$ users, $2M=6$ virtual beams and $\Tdl=2$. Here the maximum matching size is equal to two, and by omitting beams number 2 and 5 (red crosses), the MILP satisfies the constraint \eqref{eq:my_opt_1-two}, since now each user is connected to 2 ($\le \Tdl=2$) active beams. 

\subsection{Common DL Channel Training and Multi-User Precoding} \label{sec:precoding} 
Using the MILP solution, let us define $\Vc^\star=\{m:x_m^\ast=1\} := \{m_1,m_2,\ldots, m_{M'}\}$ as the set of $M'$ ``active" virtual beams (with cardinality $|\Vc^\star|=M'$) and $\Kc^\star = \{k: y^\ast_k = 1\}:= \{k_1,k_2,\ldots, k_{K'}\}$ as the set of $K'$ active users. We design the sparsifying precoding matrix in (\ref{eq:cs_eq_1}) as
\begin{equation}\label{eq:B_mat_final}
\Bm = \widetilde{\Fm}^{\herm}_{\Vc^\star},
\end{equation}
where $\widetilde{\Fm}_{\Vc^\star}$ is the $2M\times M'$ matrix consisting of the columns of $\widetilde{\Fm}$ whose indices are in $\Vc^\star$. The effective DL channel vector of user $k$ is given by the concatenation of this precoder with the full-dimensional channel, so that we have $\widetilde{\hv}_{dl,k} = \Bm \hv_{dl,k}\approx \widetilde{\Fm}^{\herm}_{\Vc^\star} \widetilde{\Fm}\gv_k,$ where the approximation is only due to the approximate virtual beam representation in \eqref{eq:h_dl_represent}. It is easy to show that, the vector $\widetilde{\hv}_{dl,k}$ is of dimension $M'$, and has significantly large components only over a subset of $\{1,2,\ldots,M'\}$ determined by the intersection of $\Vc^\star$ and the support of $\gv_{k}$, i.e. by $\Vc^\star \cap \Jc_k$. Recall that satisfying constraint \eqref{eq:my_opt_1-two} ensures that $|\Vc^\star \cap \Jc_k|\le \Tdl$, so that one can stably recover the effective channel vector by taking $\Tdl$ linearly independent pilot measurements via the matrix $\Psim$ (see Lemma \ref{lem:stable_rec}). A convenient choice is to let the DL pilot matrix $\Psim$ to be proportional to a random unitary matrix of dimension $\Tdl \times M'$, such that 
$\Psim \Psim^\herm = \Pdl  \Id_{\Tdl}$. Once user $k$ collects its pilot signal measurements in the form of the $\Tdl$-dimensional vector $\yv_{{dl},k}$, it feeds them back to the BS in $\Tdl$ UL channel uses via analog unquantized feedback (this type of feedback is analyzed in e.g. \cite{caire2010multiuser,kobayashi2011training}). Upon receiving the noisy pilot measurements $\yv_{{dl},k} = \Psim \Bm \hdlk +\zv_k$ for any user $k\in \{1,\ldots,K\}$, the BS can obtain the minimum mean squared error (MMSE) estimate of the $2M$-dimensional DP channel $\hdlk$ as
\begin{equation}\label{eq:eff_ch_est}
\widehat{\hv}_{dl,k} = \Sigmam_{\hdlk \yv_{{dl,k}}} \Sigmam_{\yv_{dl,k}\yv_{dl,k}}^{-1}\yv_{dl,k},
\end{equation}
where \scalebox{0.9}{$\Sigmam_{\hdlk \yv_{dl,k}} = \bE \left[  \hdlk \yv_{dl,k}^\herm \right] = \Sigmam_k^{dl} \Bm^\herm \Psim^\herm$} and \scalebox{0.9}{$\Sigmam_{\yv_{dl,k}\yv_{dl,k}} = \bE \left[ \yv_{dl,k} \yv_{dl,k}^\herm \right] = \Psim  \Bm  \Sigmam_k^{dl} \Bm^\herm \Psim^\herm + N_0 \mathbf{I}_{\Tdl} $}.

\subsection{Beamforming and Data Transmission}
Without loss of generality, let us assume that the BS wants to serve the first $K'$ users, using a beamforming scheme that is ideally interference-free. We consider zero-forcing beamforming (ZFBF) for this purpose, where the ZFBF matrix $\Vm_{\sf ZF}$ is given by the column-normalized 
version of the Moore-Penrose pseudoinverse of the estimated effective channel matrix defined as $\widehat{\Hm}_{\sf eff} = \Bm \widehat{\Hm} = \Bm \left[ \widehat{\hv}_{dl,1},\widehat{\hv}_{dl,2},\ldots, \widehat{\hv}_{dl,K'}\right]\in \bC^{M'\times K'}$, so that we have $\Vm_{\sf ZF} = 
\widehat{\Hm}_{\sf eff}^\dagger \Jm^{1/2}$, where $\widehat{\Hm}_{\sf eff}^\dagger = 
\widehat{\Hm}_{\sf eff}  \left ( \widehat{\Hm}_{\sf eff}^\herm \widehat{\Hm}_{\sf eff} \right )^{-1}$ and $\Jm$ is a diagonal matrix, normalizing the columns of $\Vm_{\sf ZF}$. A channel use of the DL precoded data transmission phase at the $k$-th user receiver takes on the form
\begin{equation} \label{receivedk}
r_k =  \hdlk ^\herm \Bm^\herm  \Vm_{\sf ZF} \Pm^{1/2} \sv + n_k, 
\end{equation} where $\sv \in \bC^{K' \times 1}$ is a vector of unit-energy user data symbols, $\Pm$ is a diagonal matrix defining the 
power allocation to the DL data streams and $n_k\sim\cg (\mathbf{0},N_0 )$ is the AWGN. The transmit power constraint is given by $ \trace( \Bm^\herm \Vm_{\sf ZF} \Pm \Vm^\herm_{\sf ZF} \Bm ) = \trace ( \Vm_{\sf ZF}^\herm \Vm_{\sf ZF} \Pm ) = \trace (\Pm) = \Pdl$,
where we used $\Bm \Bm^\herm = \Id_{M'}$ and the fact that $\Vm_{\sf ZF}^\herm \Vm_{\sf ZF}$ has unit diagonal elements by construction. 
We use the simple uniform power allocation $[\Pm]_{k,k} = \frac{\Pdl}{K'}$ to each $k$-th user data stream. The received symbol at user $k$ receiver is given by $ r_k = b_{k,k} \sv_k + \sum_{\ell \neq k} b_{k,\ell} \sv_{\ell}   +  n_k,$
where the coefficients $b_{k,1}, \ldots, b_{k,K'}$ are given by the elements of the $1 \times K'$ row vector $\hdlk^\herm \Bm \Vm \Pm^{1/2}$ in (\ref{receivedk}). In the presence of an accurate channel estimation we expect that $b_{k,k} \approx \sqrt{[\Jm]_{k,k} [\Pm]_{k,k}}$ and $b_{k,\ell} \approx 0$ for $\ell \neq k$. However, this is not a given, since in general there typically exists a non-negligible channel estimation error. For simplicity, in order to calculate the ergodic sum-rate, here we assume that the coefficients 
$b_{k,1}, \ldots, b_{k,K'}$ are known to the corresponding receiver $k$. Including the DL training overhead, this yields the rate expression (see \cite{caire2018ergodic}):
\begin{equation}\label{eq:rate_ub}
R_{\rm sum}  = \left (1 - \frac{\Tdl}{T} \right ) \sum_{k =1}^{K'} \bE \left [ \log \left ( 1 + \frac{\left|b_{k,k}\right|^2}{N_0 + \sum_{\ell\neq k} \left|b_{k,\ell} \right|^2} \right ) \right ].
\end{equation}	

\section{Simulation Results}\label{sec:simulations}
In this section, we empirically examine the performance of our scheme in different aspects of dual-polarized UL channel covariance estimation, UL-DL covariance transformation and common multi-user DL channel training and precoding. We compare the covariance estimation performance of our method with the sample covariance estimator in terms of the mean normalized Frobenius norm error, defined as
\begin{equation}\label{eq:E_NF}
E_\text{NF} = \mathbb{E}\left\{  \frac{||\mathbf{\Sigmam}_\mathbf{h}- \wh{\Sigmam}_\mathbf{h}||_{\sf Fß}   }{||\mathbf{\Sigmam}_\mathbf{h}||_{\sf F}}  \right\},
\end{equation}
where $\Sigmam_{\hv}$ is the true channel covariance and $\widehat{\Sigmam}_{\hv}$ is its estimate and where the expectation is taken over several sources of randomness in the channel, namely, random ASFs, random channel realizations in the sample set and random additive noise.  

We consider a BS equipped with a ULA of $M$ antennas with $\lambda_{ul}/2$ spacing. To examine the covariance estimation performance, we suppose $N=2\kappa M$ independent samples of the $2M$-dimensional dual-polarized channel are available, where $\kappa=\frac{N}{2M}$ denotes the ratio between the sample set size and the channel dimension. 
The number of density functions used to approximate the continuous DP-ASF in \eqref{eq:Gamma_c_approx} is set to $n=3 M$. In order to produce (semi-)random Horizontal and Vertical ASFs we consider the following generative model: 
\begin{equation}\label{eq:gamma_H_rnd}
\begin{aligned}
\gammaH (\xi) = \frac{\alpha}{|\Ic_1|+|\Ic_2|}\left(\rect_{\Ic_1}(\xi)+\rect_{\Ic_2}(\xi) \right)+ \frac{1-\alpha}{2}\left( \delta(\xi - \xi_1)+\delta (\xi - \xi_2) \right),
\end{aligned}
\end{equation}
where for an interval $\Ic\subset [-1,1]$, we have defined the rectangular function as $\rect_{\Ic}(\xi) =1$ for $\xi \in \Ic$ and $\rect_{\Ic}(\xi) =0$ for $\xi \notin \Ic$.
The intervals $\Ic_1$ and $\Ic_2$ are subsets of $[-1,1]$, each of length $|\Ic_1|$ and $|\Ic_2|$, respectively, where the lengths are chosen uniformly at random between $0.1$ and $0.4$, i.e. $|\Ic_j|\sim \mathbb{U}([0.1,0.4])$, independently for $j=1$ and $j=2$. Besides, $\xi_1,\, \xi_2\in [-1,1]$ denote discrete AoAs, generated independently and uniformly at random over $[-1,1]$. The scalar $\alpha \in [0,1]$ denotes what we call the continuous-to-discrete ASF ratio. Basically, since $\int_{-1}^1 \frac{1}{|\Ic_1|+|\Ic_2|}\left(\rect_{\Ic_1}(\xi)+\rect_{\Ic_2}(\xi) \right) \, d\xi=1$ and $\int_{-1}^1 \frac{1}{2}\left( \delta(\xi - \xi_1)+\delta (\xi - \xi_2) \right)\, d\xi=1$, $\alpha$ controls the contribution of the continuous part versus the discrete part to the overall ASF: for $\alpha =0$ we have a purely discrete ASF, for $\alpha = 1$ we have a purely continuous one and for $\alpha \in (0,1)$ we have a mixture of the two. Similarly, we generate the vertical ASF as: 
\begin{equation}\label{eq:gamma_V_rnd}
\begin{aligned}
\gammaV (\xi) &= \frac{\alpha}{|\Ic_1'|+|\Ic_2'|}\left(\rect_{\Ic_1'}(\xi)+\rect_{\Ic_2'}(\xi) \right)+ \frac{1-\alpha}{2}\left( \delta(\xi - \xi_1')+\delta (\xi - \xi_2') \right),
\end{aligned}
\end{equation}
Since it is natural for the horizontal and vertical ASFs to overlap in their support, we assume the discrete AoAs to be the same, i.e. $\xi_1'=\xi_1$ and $\xi_2'=\xi_2$, and we assume $\Ic_1'$ and $\Ic_2'$ to be slightly shifted versions of $\Ic_1$ and $\Ic_2$ as $\Ic_1' = \Ic_1+0.1$ and $\Ic_2' = \Ic_2+0.1$. Finally, we assume the cross-correlation function $\rho (\xi)$ to take on the form $\rho (\xi) = \beta \sqrt{\gammaH (\xi) \, \gammaV (\xi)}$, where $\beta \in [0,1]$ is a scalar that controls the cross-correlation level between H and V channels. This is a simplifying assumption on the form of $\rho (\xi)$, which does not undermine the generality of the DP-ASF, and satisfies the necessary condition $|\rho (\xi)|^2\le \gammaH (\xi) \gammaV (\xi)$ for the DP-ASF $\Gammam (\xi)$ to be a PSD matrix-valued function for all $\xi \in [-1,1]$. In addition, we can change the cross-correlation between H and V channels simply by changing $\beta$. The larger $\beta$ is, the more correlated the polar channels are.

\begin{figure*}
	\centering
	\begin{subfigure}{.5\textwidth}
		\centering
		\includegraphics[width=.9\linewidth]{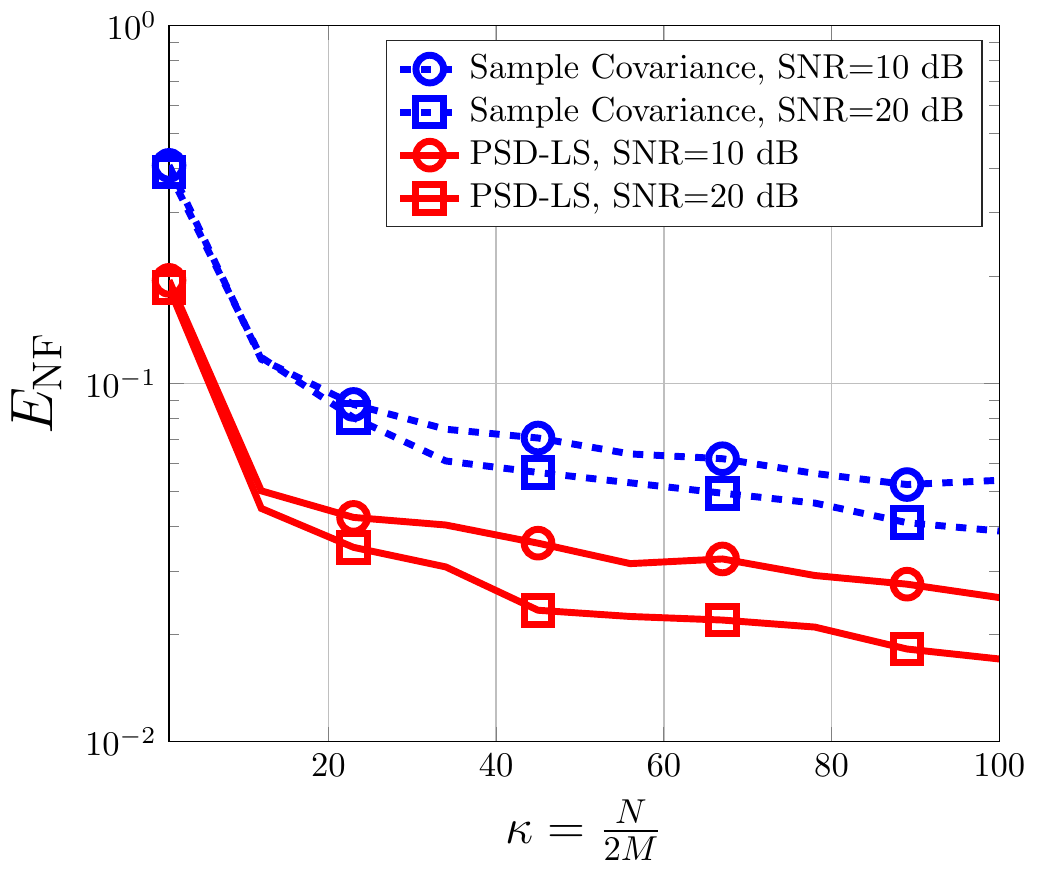}
		\label{fig:sub1}
	\end{subfigure}%
	\begin{subfigure}{.5\textwidth}
		\centering
		\includegraphics[width=.9\linewidth]{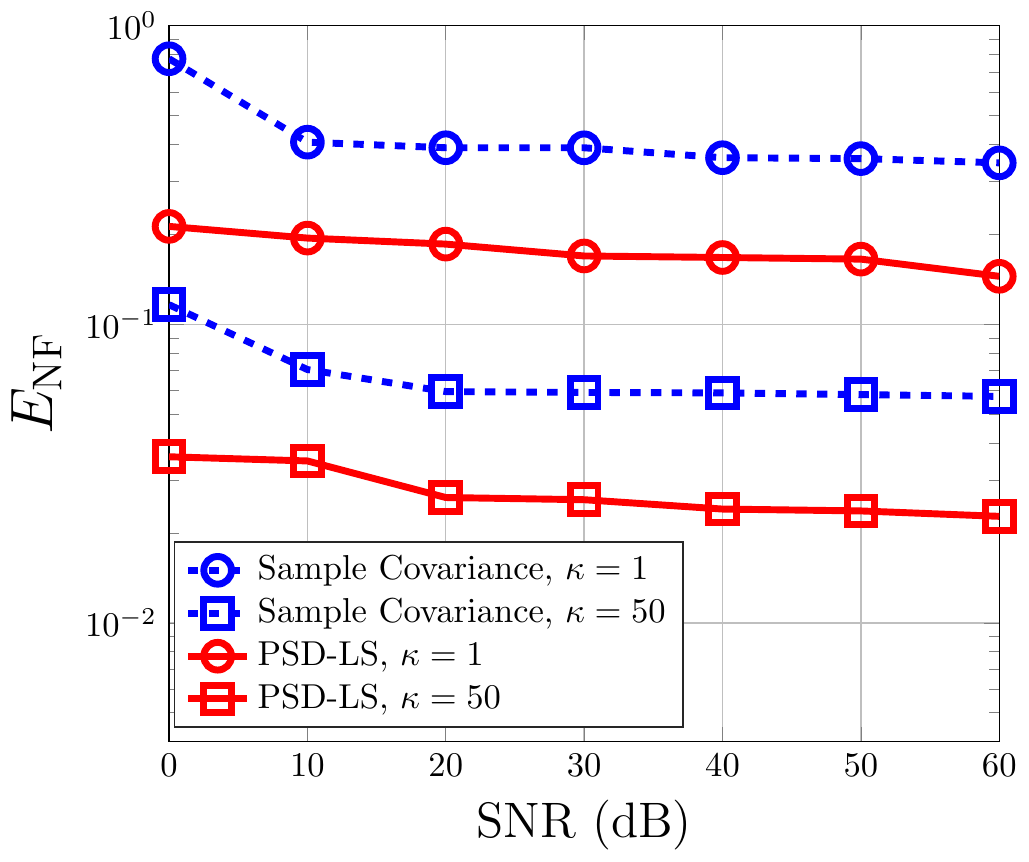}
	\end{subfigure}
	\caption{\small Channel covariance estimation error vs the sample size to channel dimension ratio (left) and SNR (right) for $M=32$.}
	\label{fig:UL_cov_est}
\end{figure*}
\subsection{UL Covariance Estimation Error}
The first experiment compares the UL covariance estimators. We consider a ULA of size $M=32$. To perform a Monte-Carlo simulation, we generate 100 random DP-ASFs according to the model explained earlier. For each random DP-ASF, we generate $N$ independent samples of the channel as $\hv_{ul}(1),\ldots,\hv_{ul}(N)$ and AWGN vectors $\zv (1),\ldots,\zv (N)$ to generate the noisy pilot signals $\yv_{ul} (i)=\hv_{ul} (i)+\zv(i),\,i=1,\ldots,N$. We repeat this for 50 different realizations of channel and noise, each time estimating the covariance given pilot signals and computing the estimation error. Therefore, the UL covariance estimation error is eventually averaged over $100 \times 50 = 5000$ random instances to empirically compute the error metric in \eqref{eq:E_NF}. Fig. \ref{fig:UL_cov_est} compares the normalized Frobenius norm error as a function of the sampling ratio (left figure) as well as the SNR (right figure). The error figures show that the method based on PSD-LS considerably improves estimation accuracy in comparison to the sample covariance estimator. The main reason is that, PSD-LS captures the structure of the dual-polarized covariance (see \eqref{eq:PSDLS}): it enforces the Kronecker structure by adopting the parametric covariance form $\sum_{i=1}^{n+\widehat{r}} \Wm_i \otimes \Sm_i$ and it constraints the coefficients $\Wm_i,~i=1,\ldots,n+\widehat{r}$ to be PSD in accordance with the DP-ASF being a PSD matrix-valued function. 

\subsection{UL-DL Covariance Transformation Error}
The second part of our proposed  scheme involves UL to DL covariance transformation as explained in Section \ref{sec:transformation}. Using the same simulation setup as introduced earlier, we study the DL covariance estimation error. In order to separately study the error of covariance transformation and that of UL covariance estimation from random channel samples, we consider two cases: in the first case we assume that the true UL covariance is given, perform the transformation and compute the error. In the second case, we assume that only the noisy pilot signals $\yv_{ul} (1),\ldots,\yv_{ul} (N)$ are given. Obviously, the estimation error is expected to be larger in the second case. Mathematically, in the first case we replace $\widehat{\Sigmam}_{\hv}^{ul}$ with $\Sigmam_{\hv}^{ul}$ in \eqref{eq:PSDLS} and estimate the ASF parametric form, whereas in the second case we compute $\widehat{\Sigmam}_{\hv}^{ul}$ as $\widehat{\Sigmam}_{\hv}^{ul}=\frac{1}{N}\sum_{i=1}^{N}\yv_{ul}  (i) \yv_{ul}  (i)^\herm - N_0 \mathbf{I}$. Finally, we also plot the error measures for UL covariance estimation from the noisy pilots to compare it to the other two other cases.
\begin{figure*}
	\centering
	\begin{subfigure}{.5\textwidth}
		\centering
		\includegraphics[width=.9\linewidth]{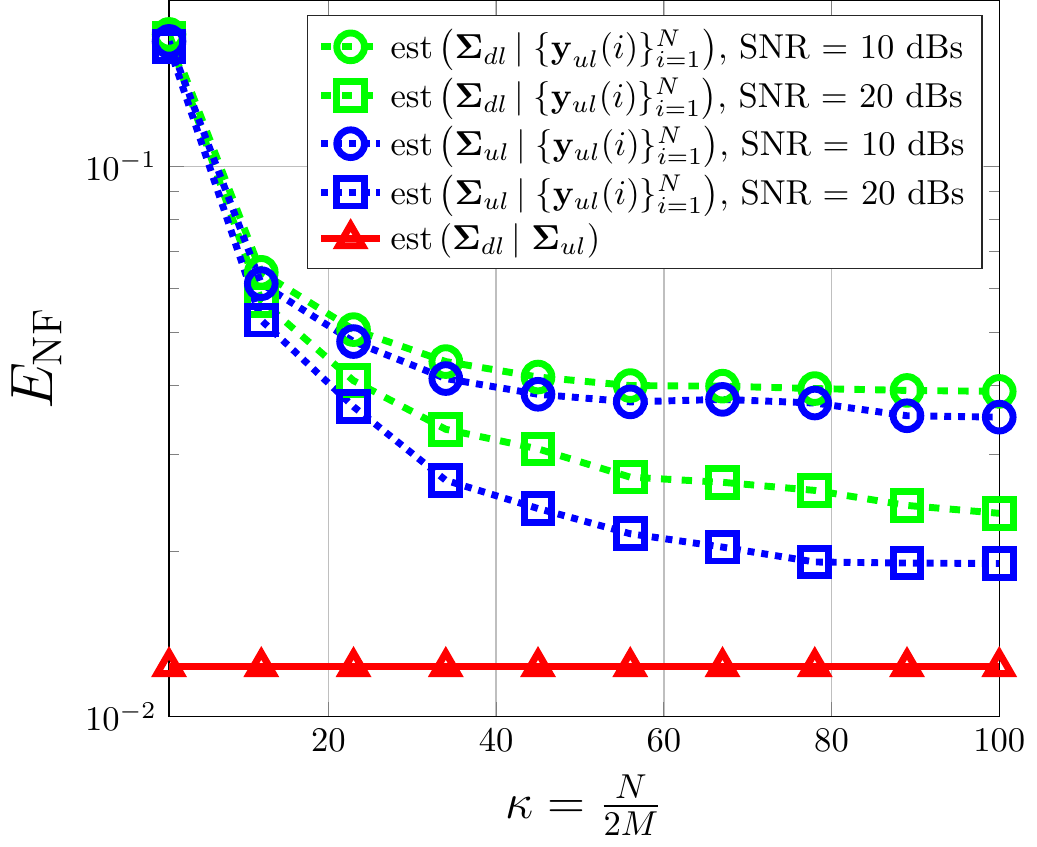}
	\end{subfigure}%
	\begin{subfigure}{.5\textwidth}
		\centering
		\includegraphics[width=.9\linewidth]{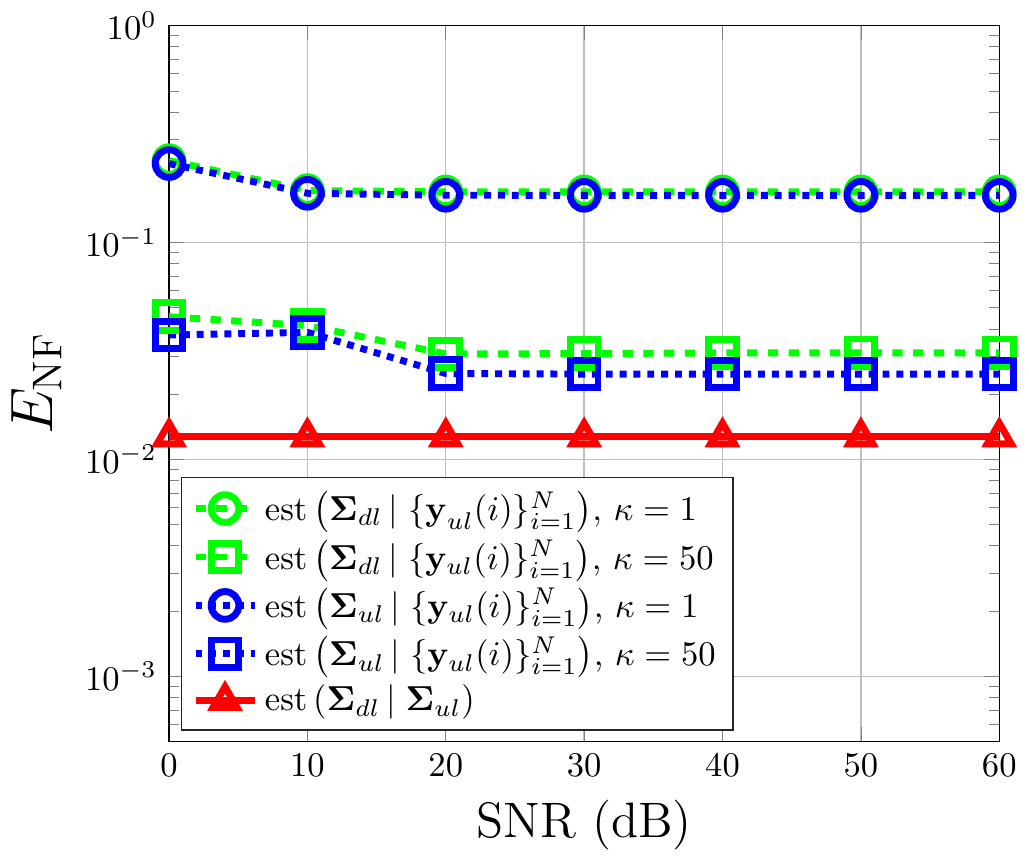}
	\end{subfigure}
	\caption{\small Channel covariance transformation error vs the sample size to channel dimension ratio (left) and SNR (right) for $M=32$. $\text{est}(\Xm | \Ym)$ denotes the estimate of $\Xm$ given $\Ym$.}
	\label{fig:cov_transform}
\end{figure*}
Fig. \ref{fig:cov_transform} illustrates the error vs sampling ratio (left figure) and error vs SNR curves (right figure). The figures show that, given a precise estimate of the UL covariance, the DL covariance can be estimated with a low error. In other words, the dominant source of error lies not in the UL-DL covariance transformation module, but in estimating the UL covariance from noisy pilots. This shows how effective the covariance transformation algorithm is. It also points to the more reasonable way of estimating the DL covariance. Collecting DL channel samples and using them to estimate the DL covariance is inefficient since it consumes too many resources to gather enough channel samples for a precise estimate of the covariance, especially since DL pilot measurements must be sent to the BS via closed-loop feedback. Instead, the BS can take in a sufficiently high number of UL channel samples, accurately estimate the UL covariance and perform UL-DL covariance estimation to obtain the DL covariance with much less error.  

\subsection{Sum-Rate Assessment of ACS}
The third part of the implementation developed in this work was dedicated to an efficient common DL channel training and multi-user precoding. For any DL pilot dimension, the ACS approach enables the BS to stably estimate effective user channel vectors while maximizing the effective channel matrix rank. In this section we present results to study the performance of ACS in terms of sum-rate, for various DL pilot dimensions and SNR values. As a multi-user scenario, we consider $K$ users, with covariances that are generated as follows. Define the four rectangular functions: $\Ic_1(\xi)=\rect_{[-0.8 , -0.6]}$, $\Ic_2(\xi)=\rect_{[-0.45 , -0.25]}$,
$\Ic_3(\xi)=\rect_{[0.1 , 0.3]}$, $\Ic_4(\xi)=\rect_{[0.5,0.7]}$. Each of these functions represents angular power density of a single scatterer in the environment. We assume that the DP-ASF components of a single generic user are (semi-)randomly generated as 

\begin{equation}
\begin{aligned}
\gammaH (\xi) &= \frac{\alpha}{Z} \left( \rect_{\Ic_i} (\xi) + \rect_{\Ic_j} (\xi)+\right)  +\frac{1-\alpha}{2}(\delta (\xi - \xi_1) +\delta (\xi - \xi_2)  ),
\end{aligned}
\end{equation}
where $i,j\in \{1,2,3,4\}$ are uniformly generated random indices, $\alpha=0.5$ is the continuous-to-discrete ASF ratio, $Z$ is a normalizing scalar such that $\int_{-1}^1\frac{1}{Z} \left( \rect_{\Ic_i} (\xi) + \rect_{\Ic_i} (\xi)\right) d\xi =1,$ and $\xi_1,\, \xi_2$ are discrete AoAs, generated independently and uniformly at random over $[-1,1]$. In order to generate the vertical ASF, similar to the previous section, we assume that the support of the continuous part of $\gammaV$ is a slightly shifted version of the support of the continuous part of $\gammaH$. We also assume that they share the same support for their discrete part. Then we have
\begin{equation}
\begin{aligned}
\gammaV (\xi) &= \frac{\alpha}{Z} \left( \rect_{\Ic_i'} (\xi) + \rect_{\Ic_j'} (\xi)+\right) +\frac{1-\alpha}{2}(\delta (\xi - \xi_1') +\delta (\xi - \xi_2')  ),
\end{aligned}
\end{equation}
where $\Ic_i'=\Ic_i+0.1$, $\Ic_j'=\Ic_j + 0.1$ and $\xi_1'=\xi_1$, $\xi_2'=\xi_2$. Besides, we suppose the H-V cross-correlation function to take on the form $\rho (\xi) = \beta \sqrt{\gammaH (\xi) \, \gammaV (\xi)}.$
\begin{figure*}
	\centering
	\begin{subfigure}{.5\textwidth}
		\centering
		\includegraphics[ width=0.92\textwidth]{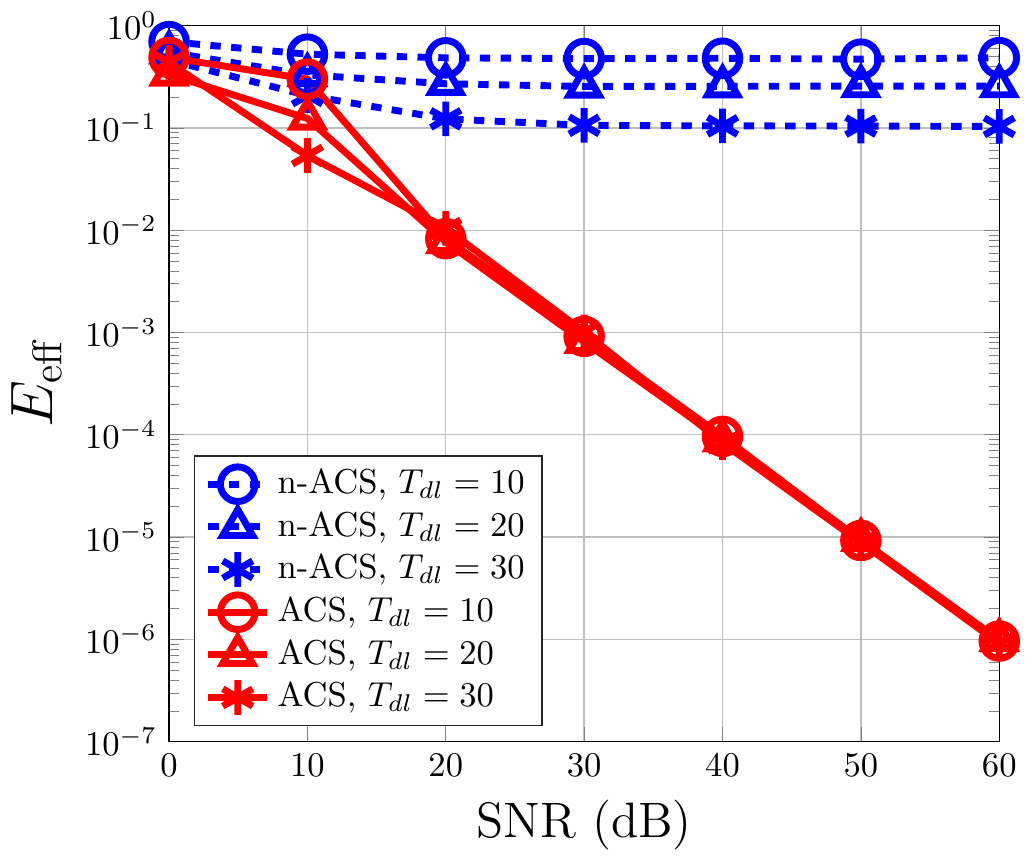}
		\caption{$M=64$, $K=8$}
		\label{fig:eff_err_M_32}
	\end{subfigure}%
	\begin{subfigure}{.5\textwidth}
		\centering
		\includegraphics[ width=0.9\textwidth]{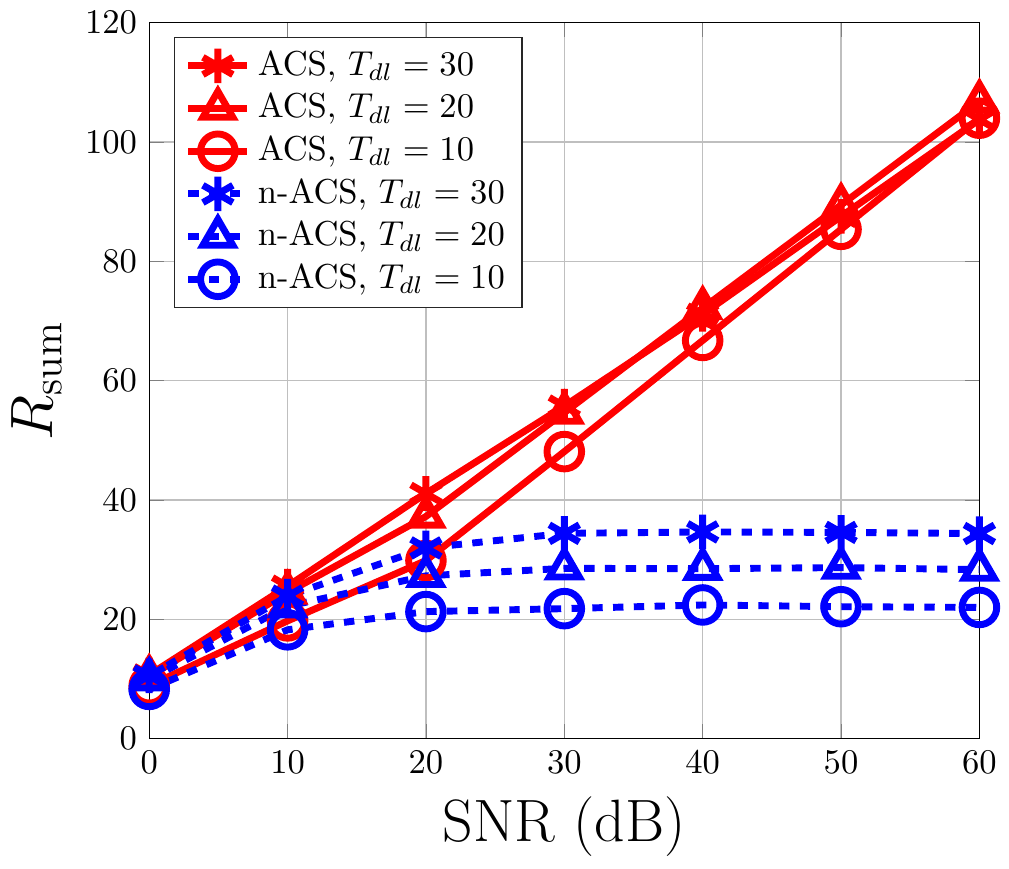}
		\caption{$M=64$, $K=8$}
		\label{fig:rate_vs_SNR_M_32}
	\end{subfigure}

	\begin{subfigure}{.5\textwidth}
		\centering
		\includegraphics[ width=0.92\textwidth]{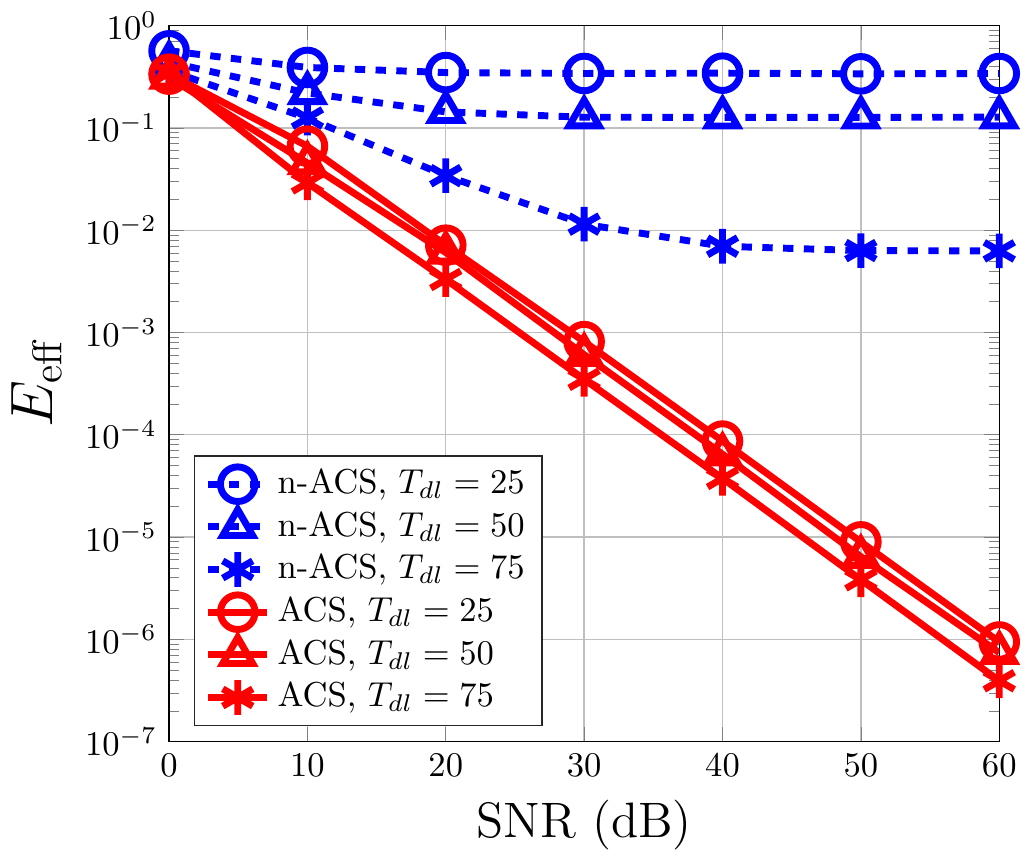}
		\caption{$M=32$, $K=6$}
		\label{fig:eff_err}
	\end{subfigure}%
	\begin{subfigure}{.5\textwidth}
		\centering
		\includegraphics[ width=0.9\textwidth]{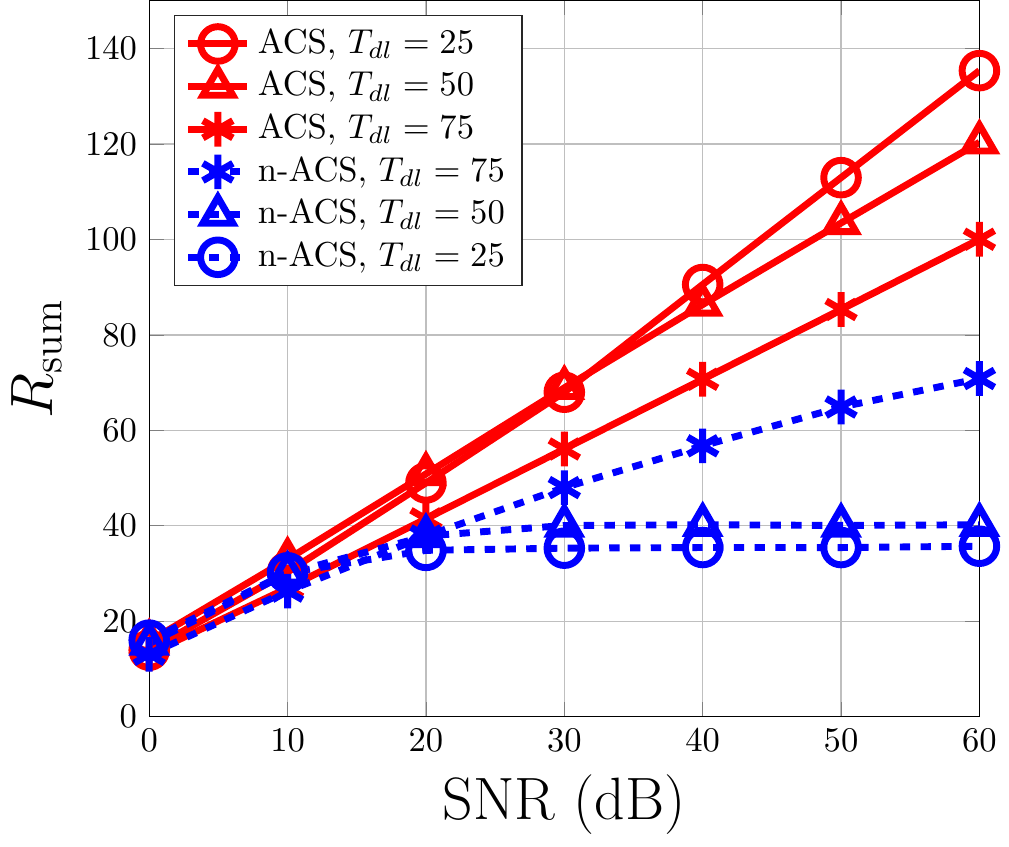}
		\caption{$M=32$, $K=6$}
		\label{fig:rate_vs_SNR}
	\end{subfigure}
	\caption{\small Effective channel estimation error vs SNR (left figures), and sum-rate vs SNR (right figures) comparison, with $M=64$ and $K=8$ for the upper row and $M=32$ and $K=6$ for the lower row. ``n-ACS" refers to the non-ACS method, implemented using random pilot vectors and no sparsification. }
	\label{fig:ACS_results}
\end{figure*}

Assuming a dual-polarized ULA with antennas, we generate semi-random DP-ASFs for $K$ users as explained above. Then we compute their covariances using \eqref{eq:ch_cov}. In order to isolate the effect of sparsification from the other parts of the implementation (UL covariance estimation, UL-DL covariance transformation), we assume that the true DL covariance for each user is available at the BS. For a given DL pilot dimension, we implement ACS by designing the DL precoder via the MILP in \eqref{opt:P_MILP_Thm} for common training and estimation of the effective channels. Next the users are served through a ZFBF scheme and the sum-rate is computed via \eqref{eq:rate_ub}. We compare the performance of ACS with that of non-ACS training. The latter case is equivalent to setting the precoder in \eqref{eq:cs_eq_1} to $\Bm = \mathbf{I}_{2M}$, i.e. \textit{not} sparsifying the channels. Apart from the sum-rate metric, we also compute the mean squared error (MSE) of estimating the effective channels via the following formula:
\begin{equation}\label{eq:E_eff}
E_{\text{eff}} = \frac{1}{|\Kc^\star|}\sum_{k\in \Kc^\star }\bE \left\{ \frac{\left\Vert \Bm \left(\hdlk - \widehat{\hv}_{dl,k} \right)\right\Vert^2}{\Vert \Bm \hdlk \Vert^2}  \right\}
\end{equation}
where we recall that $\Kc^\star$ is the set of users selected to be served by the MILP. 

\begin{figure}[t]
	\begin{subfigure}{.5\textwidth}
		\includegraphics[height=7.2cm, width=8cm]{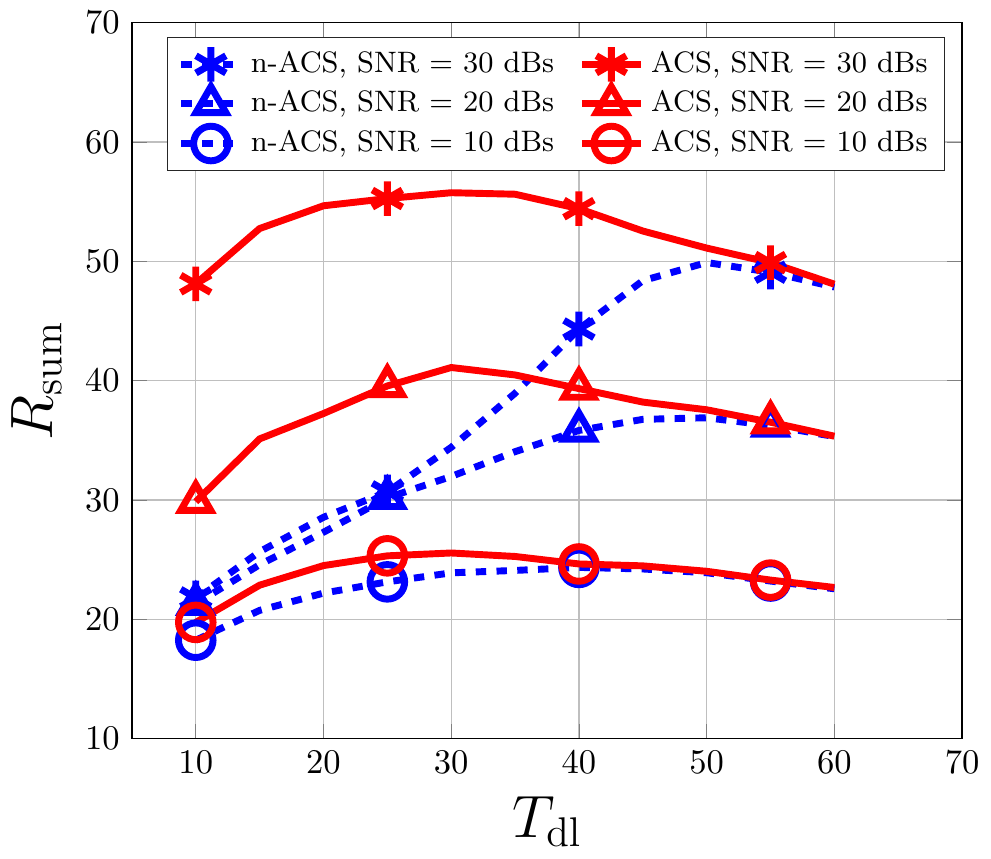}
		\caption{$M=32$, $K=6$}
		\label{fig:Rate_vs_Tdl_M_32}
	\end{subfigure}%
	\begin{subfigure}{.5\textwidth}
		\includegraphics[height=7cm, width=8cm]{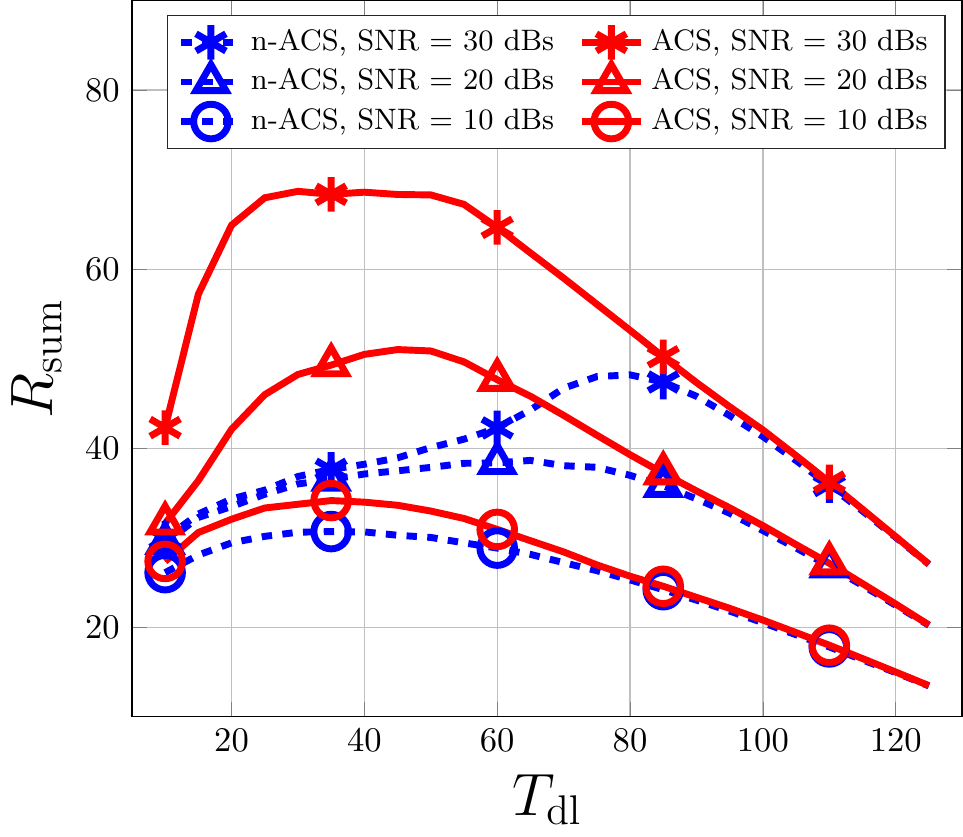}
		\caption{$M=64$, $K=8$}
		\label{fig:Rate_vs_Tdl_M_64}
	\end{subfigure}
	\caption{sum-rate vs pilot dimension curves.}
	\label{fig:rate_vs_Tdl}
\end{figure} 
See the results of Fig. \ref{fig:ACS_results}, in which we have plotted the effective error and sum-rate curves as a function of SNR for two different system setups, where in one the array size is $M=32$ and serves $K=6$ users, and in the other the array size is $M=64$ and serves $K=8$ users. In each case, we illustrate the results for different values of DL pilot dimension $\Tdl$. First, note that with the ACS method, the effective channel estimation error decreases linearly with the the increase of $\log( \text{SNR})$ (or the SNR in dBs). In contrast, with n-ACS this error is saturated to a fixed value and does not decrease by increasing the SNR. This behavior is a direct outcome of Lemma \ref{lem:stable_rec}, which states that if the pilot dimension is less than the sparsity order of the effective channel $\Bm \hv$, then the estimation error does not tend to zero with increasing the SNR. Conversely, a stable estimation is possible if the pilot dimension is larger than the sparsity order of the effective channel, which is enabled by ACS through the MILP.  
	
	Stable estimation of the effective channel is also important in achieving an interference-free DL transmission. This can be seen by comparing the sum-rate curves of the ACS and non-ACS methods (Figs. \ref{fig:rate_vs_SNR} and \ref{fig:rate_vs_SNR_M_32}) in the high-SNR regime. With the non-ACS method, the sum-rate saturates to a fixed value as SNR increases, demonstrating an interference-limited behavior. However, with ACS the sum-rate increases linearly with $\log (\text{SNR})$, achieving much higher sum-rates in the medium-to-high-SNR regime. 
	
	Fig. \ref{fig:rate_vs_Tdl} illustrates the sum-rate vs pilot dimension curves for the two setups as before and for various SNR values. The point of this figure is to show the relationship between the pilot dimension and the sum-rate. From \eqref{eq:rate_ub} we note that the pilot dimension controls a trade-off in consuming time-frequency resources: increasing the pilot dimension results in better channel estimation and therefore less interference, which increases the argument inside the logarithm in \eqref{eq:rate_ub}, but it decreases the pre-log factor $1-\Tdl/T$ and leaves fewer resources for data transmission. Therefore, we expect that there exists an optimal pilot dimension which maximizes the sum-rate for any given setup. This can be seen from the curves of Fig. \ref{fig:rate_vs_Tdl}. Note that in all setups the ACS method achieves higher sum-rates compared to the non-ACS method for the same pilot dimension, in some cases achieving almost twice the sum-rate of the non-ACS method.


\section{Conclusion}
We proposed a thorough implementation of a multi-user FDD massive MIMO system with dual-polarized antenna elements. We addressed the dimensionality challenge of such systems through a three-step process: (1) UL covariance estimation from limited, noisy UL channel samples, (2) UL-DL covariance transformation, and (3) active channel sparsification and multi-user precoding for DL channel training and interference-free beamforming. Using error and sum-rate metrics we showed that our approach is successful for implementing dual-polarized FDD massive MIMO systems, overcoming the curse of prohibitively large dimensions and limited time-frequency resources.

	{\footnotesize
		\bibliographystyle{IEEEtran}
		\bibliography{references}
	}
	
\end{document}